\documentclass[11pt]{article}
\usepackage[margin=1in]{geometry}
\usepackage[T1]{fontenc}
\usepackage{amsmath}
\usepackage{amsthm}
\usepackage{newtxtext,newtxmath}
\usepackage{graphicx,booktabs,xcolor,tikz,setspace}
\usepackage{subcaption}
\usetikzlibrary{decorations.pathreplacing}
\usetikzlibrary{patterns}
\usepackage{natbib}
\bibpunct[, ]{(}{)}{,}{a}{}{,}
\usepackage[hidelinks]{hyperref}
\usepackage[capitalise]{cleveref}
\crefname{appendix}{Appendix}{Appendices}
\crefname{assumption}{Assumption}{Assumptions}

\theoremstyle{plain}
\newtheorem{proposition}{Proposition}
\newtheorem{lemma}{Lemma}
\newtheorem{corollary}{Corollary}
\newtheorem{assumption}{Assumption}
\theoremstyle{remark}
\newtheorem{remark}{Remark}

\renewcommand{\proof}[1]{\par\medskip\noindent\textit{#1}\enspace\ignorespaces}
\renewcommand{\endproof}{\par\medskip}
\newcommand{\Halmos}{\hfill$\square$}
\newcommand{\ACKNOWLEDGMENT}[1]{\section*{Acknowledgments}#1}

\newcommand{\ECSwitch}{\clearpage
  \setcounter{section}{0}\setcounter{equation}{0}\setcounter{figure}{0}%
  \setcounter{table}{0}\setcounter{lemma}{0}%
  \renewcommand{\thesection}{EC.\arabic{section}}%
  \renewcommand{\theequation}{EC.\arabic{equation}}%
  \renewcommand{\thefigure}{EC.\arabic{figure}}%
  \renewcommand{\thetable}{EC.\arabic{table}}%
  \renewcommand{\thelemma}{EC.\arabic{lemma}}}
\newcommand{\ECHead}[1]{\begin{center}{\Large\bfseries E-Companion: #1}\end{center}}
\title{Capital Costs Protect the Network: Mining Gap Dynamics in Proof-of-Work Blockchains Under a Stochastic Mempool}
\author{Hasan Selçuk Kılıç\\ \small Department of Electrical and Electronics Engineering, Bilkent University, 06800 Bilkent, Ankara, Turkey\\ \small \texttt{selcuk.kilic@bilkent.edu.tr} \and Nail Akar\\ \small Department of Electrical and Electronics Engineering, Bilkent University, 06800 Bilkent, Ankara, Turkey\\ \small \texttt{akar@ee.bilkent.edu.tr}}
\date{October 2026}
\begin{document}
\maketitle
\begin{abstract}
Proof-of-Work blockchains increasingly rely on transaction fees to keep miners engaged
as block subsidies decline. An immediate outcome of this is that miners may sit idle
at the start of a round until enough high-fee transactions accumulate
in order to make mining worthwhile. This idle period is known as the mining gap. We model the high-fee transaction backlog as a Lindley-type
recursion under a two-tier fee market and use its stationary distribution to
characterize the mining gap. Under Bitcoin's difficulty-adjustment mechanism, the expected
gap grows strictly with the hashpower active on the network and consumes the entire
round as that hashpower grows without bound. A free-entry condition
determines a unique long-run equilibrium hashpower. Capital cost caps the
equilibrium hashpower and prevents the mining gap from occupying the entire
round. We further characterize how the block subsidy, transaction fees, and
operating cost shape the equilibrium hashpower and the mining gap.

\medskip\noindent\textbf{Key words:} proof-of-work, mining gap, blockchain security, transaction fees, capital costs

\noindent\textbf{Funding:} The authors received no funding for this research.
\end{abstract}

\section{Introduction}
\label{sec:intro}
Since its introduction by Nakamoto, Bitcoin's Proof-of-Work (PoW) consensus mechanism
has underpinned trust in a fully decentralized and adversarial environment \citep{nakamoto2008}. The blockchain
is a shared ledger that records all confirmed transactions in a chain of blocks. The
participants who add new blocks to this chain are called \emph{miners}. Miners compete to
extend the blockchain by solving cryptographic puzzles. The miner who finds a block is
rewarded with \emph{block subsidies} (newly minted coins) and the \emph{transaction fees}
that users pay to have their transactions included in that block
\citep{nakamoto2008,bitcoinDevBlockChain}.
Each attempt to solve the cryptographic puzzle succeeds with a small probability. The number of attempts a
miner makes per unit of time is called its \emph{hashpower}, and a miner with more
hashpower is more likely to find the next block. Every 2016 blocks, Bitcoin's
\emph{difficulty-adjustment mechanism} makes the puzzle harder or easier so that blocks
keep arriving about every ten minutes on average \citep{bitcoinDevBlockChain}. We call
the time between two consecutive blocks a \emph{round}.

For most of Bitcoin's history, the
block subsidy has dominated miner revenue. However, the halving schedule cuts the block
subsidy in half roughly every four years \citep{carlsten2016,bitcoinDevBlockChain}. This results in a shift toward a \emph{fee-dominated
regime}, where transaction fees become the primary source of miner revenue \citep{carlsten2016,easley2019fees}. This transition
raises concerns about the security of the network, as miners may be less incentivized to
participate if transaction fees are insufficient to cover their costs \citep{capponi2023mining}. Miners face two
kinds of costs \citep{carlsten2016,budish2025trust,kimSeo2026mind}. \emph{Operating costs} are
the costs of running the mining hardware, such as electricity, and miners pay them only
while they mine. \emph{Capital costs} are the costs of owning the hardware, and miners pay them
whether they mine or stay idle. The umbrella term
for the problem of whether fees will cover these costs once the subsidy shrinks is ``security budget'', representing the total budget for incentivizing
miners to secure the network \citep{ciaian2021security,budish2025trust}. There are active debates about whether this security budget will be sufficient, and
the security budget argument is frequently used to justify protocol or monetary-policy
proposals \citep{ciaian2021security,budish2025trust}. ``Tail emission'' is one policy aimed directly at the security budget problem by
providing a nonzero block subsidy indefinitely \citep{moneroTailEmission}. This policy is implemented in cryptocurrencies such as Monero and Dogecoin
\citep{moneroTailEmission,dogecoinFAQ}, and has also been discussed in the
Bitcoin context as a possible hard-fork change to the subsidy schedule
\citep{todd2022tailEmission}. This would depart from Nakamoto's original
fee-transition design, which anticipated that the mining incentive could
eventually transition to transaction fees once the coin supply reached its
predetermined limit \citep{nakamoto2008,nakamoto2008fees}.

One concern that arises in a fee-dominated regime that is related to the security budget
is the \emph{mining gap}, an idle period at the start of each round during
which accumulated transaction fees are insufficient to cover miners' operating costs,
making mining temporarily unprofitable. The gap forms at the start of the round
because each newly found block removes the highest-fee waiting transactions, so the
next block initially has little fee revenue to offer. We call a miner \emph{honest} if
it always extends the longest chain and publishes each block as soon as it finds one, as
Bitcoin's protocol prescribes \citep{nakamoto2008,carlsten2016}. An honest miner still
decides when to run its hardware. During this gap, all honest miners are expected
to be idle. An idle period might look harmless, or even desirable, since idle miners use
no electricity. From a security standpoint, however, this is a serious concern. The work
of honest miners is what protects a PoW blockchain from attackers. To reverse a
confirmed payment, an attacker must build a competing chain of blocks that grows faster
than the chain extended by honest miners \citep{nakamoto2008}. Even while honest miners
work, the honest chain grows only when one of them solves the puzzle, and this can take
some time. When they stop working, the honest chain is guaranteed to stop growing, and
the attacker can gain ground in this race without any competition. The electricity that
honest miners save during the gap is therefore protection that the network gives up \citep{goren2019mind}. Nakamoto's original security
analysis assumes that miners are always active \citep{nakamoto2008}. A mining gap creates
an advantage for adversaries that are not bound by the same economic incentives. An
adversary who profits from the attack itself, for example by reversing a payment, has no
reason to wait for fees to accumulate. Such
adversaries can start mining immediately after a block is found, while honest miners
wait for the gap to end. In this manner, even without a majority of the hashpower,
adversaries can create problems such as double-spending attacks, selfish mining and fee
sniping \citep{nakamoto2008,rosenfeld2014double,eyal2018majority,bitcoinOptechFeeSniping}.
A double-spending attacker reverses a payment after the recipient has accepted it, so
a payment that looked settled is undone. A selfish miner withholds the blocks it finds.
This wastes the work of honest miners and earns the attacker more than its share of
hashpower would justify, which draws other miners to join the attacker and concentrates
control of the network \citep{eyal2018majority}. A fee sniper re-mines a recent block
to take its fees and thereby replaces a block that the network had already accepted
\citep{bitcoinOptechFeeSniping}.

Unconfirmed transactions wait in a pool called the \emph{mempool} until a block includes
them, and transactions that do not fit into one block wait there for the next. The fees
available at the start of a round therefore depend on how long the earlier rounds
lasted. Since the time needed to find a block is random, these fees are random as well.
\citet{carlsten2016} identified the mining gap through a simulation that keeps this
carryover. The analytical models that followed assume instead that the fees available
to miners reset to zero after every block \citep{tsabary2018,noguchi2025,kimSeo2026mind}.
None of these studies derives the length of the mining gap when fees carry over between
rounds. Security analyses of idle honest miners, such as \citet{doger2024} and
\citet{doger2025}, need this length but take it as a given input.
This paper focuses on the mining gap problem with a stochastic model of the mempool
that keeps this carryover. We derive the probability distribution of the length of the
mining gap in this model and study how the gap responds to the aggregate hashpower and to
the economic parameters of the network. Our contributions are as follows.
\begin{itemize}
\item We formalize the mining decision as a profit rate maximization problem over a
long horizon. The miners form a \emph{non-atomic} population, meaning that each miner holds a
negligible share of the aggregate hashpower. We show that the myopic threshold-based rule
is a weakly dominant strategy for every miner. The rule is myopic because it activates
a miner as soon as the current reward rate covers the operating cost, without regard
to future rounds. It is weakly dominant because no other strategy earns the miner a
higher long-run profit rate, whatever strategies the other miners use.
\item We derive the closed-form stationary distribution of the mining gap under a persistent stochastic mempool and a two-tier fee market.
The stationary distribution gives the long-run distribution of the gap length within
each \emph{epoch}, which is the run of consecutive rounds between two difficulty adjustments. In the two-tier fee market, high-fee transactions arrive over time, and
low-fee transactions are always available to fill the rest of a block. We model the high-fee transaction backlog with a Lindley-type recursion with a floor. The recursion tracks how unconfirmed
high-fee transactions carry over from one block to the next, and the floor captures that
mining does not start until
enough of them have accumulated. We first formulate its transient round-by-round distribution. Subsequently, from the mempool
distribution, we derive the stationary distribution of the mining gap and its
expectation for each winning probability. The \emph{winning probability} is the probability
that a single attempt at the puzzle succeeds, and Bitcoin's difficulty-adjustment
mechanism tunes it to keep blocks arriving at a target pace.
\item The difficulty-adjustment mechanism stops changing the winning probability once
the average round length equals its target. Following \citet{kimSeo2026mind}, we
characterize this \emph{invariant winning probability} under our mempool model and analyze its
stability. The invariant winning probability is stable when the winning probability
returns to it after a small disturbance.
\item We show that the expected mining gap is strictly
increasing in the aggregate hashpower under the stochastic mempool model. When
hashpower grows, the difficulty-adjustment mechanism makes the puzzle harder so that
blocks still arrive at the target pace. At a given level of fees, each unit of hashpower
then earns less per unit of time, so miners wait for more fees before they start. As the
aggregate hashpower grows without bound, the expected gap approaches the full length of
a round, and an attacker who mines continuously needs an ever smaller share of the
aggregate hashpower to outpace the honest chain. \citet{carlsten2016} reach a similar
conclusion in a simplified model in which no fees carry over between blocks.
\item We analyze the long-run equilibrium hashpower, the no-gap
condition, and the parameter effects on network security under the same model. By
parameter effects we mean the direction in which each security metric moves when a
single parameter, such as the block subsidy, increases while the others stay fixed.
Since miners keep joining the network only while their revenue covers both of their
costs, a positive capital cost keeps the equilibrium hashpower finite and stops the gap
from taking up the whole round. A larger block subsidy shortens the gap and raises the
equilibrium hashpower.
\item We estimate the inputs of the model from Bitcoin data of 2026 and follow the model
along the halving schedule. With free entry and the estimated costs, a gap forms only after
several more halvings, in 2044 or 2056 in our scenarios, or not before 2070 if fees fall in
bitcoin. It lasts about two and a half minutes per ten-minute round by 2070, and the
invariant winning probability stays stable. If no miner could enter or leave the network
after 2026, a gap would instead form after the next one or two halvings, unless the bitcoin
price rose enough to offset them, and the invariant winning probability would be unstable
for many years. A higher ratio of operating to capital cost, for example from cheaper
mining hardware, brings the gap forward, lengthens it and can make the invariant winning
probability unstable.
\end{itemize}

The rest of the paper is organized as follows. \Cref{sec:related} reviews related work.
\Cref{sec:model} describes the mining process and our model, and \cref{sec:objective}
derives the myopic activation rule. \Cref{sec:aggregate} derives the aggregate behavior
of the miners, and \cref{sec:backlog} derives the stationary distributions of the
backlog and the mining gap. \Cref{sec:daa} analyzes the difficulty-adjustment mechanism
and its invariant winning probability. \Cref{sec:metrics} defines the security metrics.
\Cref{sec:hashpower} shows how the mining gap grows with the aggregate hashpower, and
\cref{sec:entry} determines the equilibrium hashpower when miners can freely join or
leave the network. \Cref{sec:prevent,sec:reduce} study how the parameters prevent and
reduce the mining gap, and \cref{sec:numerical} illustrates the results with Bitcoin data.
\Cref{sec:discussion} discusses the scope of the results, and
\cref{sec:conclusion} concludes. Proofs are in the e-companion.
\section{Related Work}
\label{sec:related}
The literature most directly relevant to this paper is summarized in four parts:
game-theoretic models of the mining gap, queueing-based models of mempool dynamics,
security analysis with regard to miner inactivity, and non-atomic games as a
modeling framework.
\subsection{Gap Games}
\citet{carlsten2016} identify the mining gap and study it through
simulation. They show that once fees replace the block subsidy, a miner has an
incentive to leave part of the available fees unclaimed in a block. This leftover acts
as bait, because it makes extending that block more attractive to the next miner than forking
away from it to claim the fees outright. They call this behavior undercutting. Miners
also rationally delay rig activation when the available fee pool is insufficient to
cover operating costs. Their objective is a miner's profit accumulated over an epoch,
divided by the total time elapsed over that epoch. An epoch is a set of consecutive
rounds sharing the same difficulty level. The difficulty level is inversely
proportional to the winning probability, so the rounds of an epoch also share the same
winning probability. This differs from the profit-per-round objective used
elsewhere in the literature. Their simulation also allows fees left in the mempool at
the end of one round to carry over into the next round. A round's backlog is
therefore not reset to zero at each block.

\citet{tsabary2018} instead formalize the problem as a one-shot
\emph{gap game}. Each miner chooses a single activation time to maximize expected
revenue per round. The fee available to the winner grows deterministically from zero
at the start of that round. It has no dependence on rounds that came before. They
compute Nash equilibria for this game numerically. A key finding is that miners of
different sizes choose different gap lengths and that coalition formation among
miners is incentivized.

\citet{noguchi2025} and \citet{kimSeo2026mind} work within this
same per-round setting. Noguchi gives a closed-form equilibrium analysis for symmetric
$n$-miner gap games, handling the resulting non-differentiable payoffs through a
region-wide decomposition. Kim and Seo instead formulate an infinite-horizon optimal
control problem for each miner and establish an open-loop Nash equilibrium. Neither
treatment carries a fee backlog across blocks. A new round begins with no memory of
the fees left over from the previous one.

Two simplifications recur across this line of work. Our model departs from both. The
per-round objective, used by everyone except \citet{carlsten2016},
does not correspond to what a miner who can choose to wait across many rounds is
actually maximizing. The epoch-level profit-per-time objective used in that simulation
improves on this, and we place it on an analytical footing here.
Separately, the assumption of a mempool that resets after every block, used by
everyone except \citet{carlsten2016}, discards the stochastic backlog
of unconfirmed transactions carried over from prior rounds. We retain this backlog, as
\citet{carlsten2016} do, but derive its dynamics and stationary
behavior in closed form rather than through simulation.

Other work proposes changes to the protocol that shorten the gap.
\citet{shi2020capthegap} design an incentive mechanism that induces miners to switch on
their rigs earlier, and \citet{budinsky2023feeredistribution} pass part of each block's
fees to the miners of later blocks through smart contracts. We take the current protocol
as given and characterize the gap that it produces.
\subsection{Stochastic Mempool and Queueing Models}
\citet{kasahara2017} model Bitcoin's transaction confirmation
process as a two-priority bulk-service queue. They derive mean confirmation times for
high- and low-priority transactions. Their framework captures the two-tier fee
structure explicitly and was an important step toward quantitative mempool analysis.
It does not model miner incentives, gap formation, or the inter-round backlog
dynamics that arise when miners are selective about activation. Our model differs
from that work by tracking the high-fee backlog with possible mining gaps explicitly
across rounds as a Lindley-type recursion. We derive its stationary distribution in
closed form and obtain the gap distribution directly from it.
The rule that miners start only once enough high-fee transactions have accumulated
resembles the N-policy of queueing theory, under which an idle server is switched on
when the queue reaches a threshold
\citep{yadin1963removable,heyman1968optimal,teghem1986control}. Together with the
block capacity, it also resembles the general bulk service rule, under which service
starts only when enough customers wait and each service clears at most a fixed batch
\citep{neuts1967bulk,powell1986bulk}. In those models the threshold is a parameter of
the service policy. In our model it follows from the miners' activation rule and moves
with the winning probability that the difficulty-adjustment mechanism sets.
\subsection{Security Under Fee Regimes and Miner Inactivity}
\citet{doger2024} study the security-latency problem under general
random network delays. They extend the framework to settings where honest miners
require a minimum number of high-priority transactions before activating, formalizing
the security cost of idle phases. Their mempool assumption, that fees
reset to zero after every block, is relaxed in our model. \citet{doger2025} analyze double-spending probabilities under time-varying honest
mining rates using ruin theory. They treat the mining gap as one source of rate
variation. Both papers take the gap duration as an input to their security bounds
rather than deriving it from first principles. The closed-form expected gap we derive
supplies this missing input. On the adversarial strategy side, \citet{sarenche2024} show via reinforcement learning that the transition to a
fee-dominated reward regime can make selfish mining immediately profitable. This holds
without the two-week loss period that deters it under fixed block subsidies.
\subsection{Non-Atomic Games}
The formal treatment of populations of individually negligible agents originates with
\citet{aumann1964}. He models a market with a continuum of traders.
\citet{schmeidler1973} develops this into a general equilibrium theory of
games with a continuum of players. Within the mining literature,
\citet{carlsten2016} already appeal to a notion of non-atomic miners when
describing the limiting behavior of a large population in their simulation. More
recently, \citet{liReppenSircar2024} formulate cryptocurrency
mining as a mean-field game. Their model also represents miners as a continuum, with
each miner's own contribution to the population's aggregate hashpower treated as
negligible.
\section{Background and Model}
\label{sec:model}
\subsection{Blockchain and the Mining Process}
A blockchain is a distributed ledger that records transactions in a sequence of
blocks \citep{nakamoto2008}. No central authority governs the ledger. Instead, a decentralized set of
participants called miners compete to append the next block. Each miner
collects a batch of unconfirmed transactions, assembles them into a candidate block,
and repeatedly attempts to find a solution to a cryptographic puzzle. This process is
known as PoW \citep{nakamoto2008,bitcoinDevBlockChain}. The puzzle requires a miner to find a value such
that hashing it together with the block header produces an output below a
protocol-specified target. There is no shortcut to finding such a value, meaning that a miner must
try many hash attempts, and each attempt succeeds with a small probability. We call
this probability the winning probability and denote it by $w$. A lower winning probability means
a harder puzzle, since more attempts are then needed on average to find a valid
solution. The protocol expresses the hardness of the puzzle as a difficulty level,
which is inversely proportional to $w$. \citet{kimSeo2026mind} call the same quantity
the winning rate. We use the term probability because $w$ refers to a single hash
attempt and does not depend on time. The rate
at which a miner can perform these attempts is called the miner's hashpower.
A miner with hashpower $h$ therefore finds blocks at rate $wh$ per unit of time.
The miner who finds a valid solution first
broadcasts the block to the network. The remaining participants then verify the solution and
append the new block to their local copy of the chain.
\subsection{Epochs and the Difficulty-Adjustment Mechanism}
The protocol recalibrates the winning probability periodically rather than continuously. In
Bitcoin, recalibration occurs every $E=2016$ blocks, roughly two weeks
\citep{bitcoinDevBlockChain}. We call a run of $E$ consecutive rounds sharing a common
winning probability an epoch. We index epochs by $j=1,2,\dots$ and index rounds within
an epoch by $n=1,\dots,E$, so round $n$ of epoch $j$ has winning probability $w^{(j)}$,
constant across $n$.

At the end of epoch $j$, the protocol measures the realized average round length over
that epoch,
\begin{equation}
\bar T^{(j)} \;:=\; \frac{1}{E}\sum_{n=1}^E T^{(j,n)},
\label{eq:epochavg}
\end{equation}
where $T^{(j,n)}$ is the realized length of round $n$ of epoch $j$, and sets the next
epoch's winning probability to
\begin{equation}
w^{(j+1)} \;=\; w^{(j)}\cdot\mathrm{clamp}_{[1/4,4]}\!\left(\frac{\bar T^{(j)}}{B}\right),
\label{eq:daaprotocol}
\end{equation}
where $B$ is the protocol's target round length, ten minutes for Bitcoin
\citep{bitcoinDevBlockChain}, and the clamp restricts any single adjustment to at most
a fourfold increase or decrease. If the previous epoch ran faster than target, the
winning probability falls and the puzzle becomes harder. If it ran slower, the winning probability
rises.

Equation \eqref{eq:daaprotocol} is our model of Bitcoin's difficulty-adjustment
mechanism, the rule that is in force in the Bitcoin protocol today, and every result
in this paper is derived under it. \citet{kimSeo2026mind} model the mechanism in the
same way. \citet{tsabary2018} and \citet{noguchi2025} also require the winning
probability to settle at the value where the expected round length equals the
protocol's target. For brevity, we refer to this rule as the difficulty-adjustment
mechanism from here on.
\subsection{The Player Space}
\label{sec:playerspace}
We model the population of participants supplying hashpower to the network as a
continuum of individually negligible units. This choice reflects the scale of
Bitcoin mining, where the total hashpower is spread over a very large number of
mining machines and each machine contributes a negligible share. Representing
participants as a continuum is the approach of non-atomic games
\citep{aumann1964,schmeidler1973}, and its use in the mining literature is reviewed in
\cref{sec:related}. Let $I=[0,1]$ with $\mathcal I$ the Borel subsets of $I$, and let
\begin{equation}
(I,\mathcal I,\eta)
\label{eq:playerspace}
\end{equation}
be a non-atomic measure space. Let $\eta$ be a finite measure on $(I,\mathcal I)$
with total mass
\begin{equation}
\eta(I) \;=\; \gamma,
\label{eq:totalmass}
\end{equation}
where $\gamma>0$ is the total hashpower available to the network. We also call $\gamma$
the \emph{aggregate hashpower}, since it adds up the hashpower of all honest participants. The interval
$I=[0,1]$ is an index set. Nothing in the model depends on its length. A subset $A\subseteq I$
represents a group of participants, and $\eta(A)$ is the hashpower that group
controls collectively. Readers unfamiliar with measure theory can read $\eta$ as a
rule that assigns to every group of participants the total hashpower of that group.
The Borel subsets are the groups for which this total is well defined, and every
group that arises in our analysis is one of them. Accordingly, an integral
$\int_A f(i)\,d\eta(i)$ adds up a quantity $f(i)$ over all participants in $A$, with
each participant weighted by its hashpower.

Non-atomicity of $\eta$ gives $\eta(\{i\})=0$ for every $i\in I$. This equality holds
exactly, for every $i$ individually and for every $\gamma$. No participant carries
positive weight and none can move a quantity computed as an integral over $I$ by any
positive amount.
Read informally, this is the mathematical statement that ``no individual participant
can manipulate the network, and every miner is negligible on their own.'' Only the
aggregate distribution of hashpower across the continuum affects outcomes. A single
participant's decision, taken alone, changes nothing that the model tracks.

The non-atomic population is a deliberate restriction. It excludes participants who
hold a non-negligible share of the hashpower, such as large mining pools. Such
participants can profit from strategies that a negligible participant cannot use,
for example withholding their own blocks in selfish mining \citep{eyal2018majority}, or
slowing down their mining to lower the difficulty and then exploiting the easier puzzle
\citep{fiat2019energy}.
As \citet{carlsten2016} note, real miners are atomic, but the non-atomic model
isolates the behavior that does not depend on a miner's size. Finite-player models
of the mining gap cover large participants \citep{tsabary2018,noguchi2025,kimSeo2026mind},
and a model that combines negligible participants with a few large ones is left for
future research.
\subsection{The Mempool}
\label{sec:mempool}
Before a transaction is confirmed into a block, it waits in a buffer called the
memory pool, or mempool \citep{budish2025trust}. Miners inspect the mempool and choose which
transactions to include in their candidate block. Each block has a fixed capacity
$s_b$, measured in transaction units. When the mempool is congested and more
transactions are waiting than a single block can hold, miners must select among them.

We restrict attention to the network in a healthy state and assume no forking occurs,
and the results we derive are meant to hold within that regime. A block, once found, is immediately
adopted by the network, and no participant can remove it by extending a rival chain.
A fork occurs when two competing blocks extend the same chain, for example because a
new block takes time to reach every participant and part of the network briefly
builds on a different block \citep{decker2013propagation}. One of the competing blocks is eventually discarded,
and the work spent on it is wasted. We exclude forks deliberately, and the analytical
models of the mining gap make the same restriction. \citet{tsabary2018} let all
miners move on to the next block as soon as one is found, and \citet{kimSeo2026mind}
disregard propagation delays so that all miners observe the same mempool. The
simulation study of \citet{carlsten2016} is an exception. It shows that when fees
dominate miner revenue, miners may fork a block on purpose to claim the fees it
contains. We view the absence of forks as the favorable case for the network, since
forks waste honest work and create further incentives to attack, so we expect the
problems identified here to persist when forks occur. A model of the mining gap that
includes forks is left for future research.

A single participant $i\in I$ is an infinitesimally small share of the network's
hashpower. Given that $i$ is the one confirmed in some round, taking the available
fee now secures it outright. Waiting to collect it later is not a rational action for
a participant this small. Because $i$'s hashpower is infinitesimal, the probability
that $i$ itself is the one who eventually mines the block containing that fee is
zero, however the winning probability evolves. \citet{carlsten2016} make the same
observation about non-atomic miners, who care only about the block they have found
because they will never find another one. We
therefore assume every candidate block is filled with the highest-fee transactions
available, up to capacity $s_b$, consistent with observed miner behavior
\citep{bitcoinDevTransactions}. \citet{kimSeo2026mind} fill blocks in the same way.

We envision two classes of transactions. \emph{Type-1 transactions} carry a high fee $r_1$
per unit and arrive to the mempool at a constant rate $\lambda$. \emph{Type-0
transactions} carry a lower fee $r_0 < r_1$ per unit and are assumed to exist in an
infinite reservoir. Whenever block capacity remains after including all available Type-1
transactions, a miner fills the rest with Type-0 transactions.

Reducing the mempool to two fee tiers is a common simplification.
Priority-queue models of Bitcoin's confirmation process adopt the same reduction.
\citet{kasahara2017} split pending transactions into a small
number of priority classes to keep the confirmation process tractable. Fee-estimation
tools built on live mempool data follow a comparable practice. They quote a handful of
discrete priority tiers rather than a continuum of prices~\citep{mempoolspace2026}.
\citet{doger2024} use the same two-tier structure in their security analysis of idle
honest miners. In their model, high-priority transactions arrive over time and
low-fee transactions fill the rest of each block. Our two-tier structure is this same
reduction taken to its simplest case. It does not capture every individual
transaction's fee. It does, however, lead to an analytically tractable model. A
mempool with many fee levels and strategic fee bidding \citep{huberman2021monopoly}
would require a considerably
more complex analysis, which we leave for future research.

We let $X^{(j,n)}$ denote the number of Type-1 transactions that carry over from round
$n-1$ of epoch $j$ into round $n$ of epoch $j$, meaning that they were not included in
the previous block. If $n=1$, the carryover instead comes from round $E$ of epoch
$j-1$. There is no reset at an epoch boundary, since the mempool does not reset when
the protocol recalibrates $w$. We call $X^{(j,n)}$ the \emph{Type-1 backlog} at the
start of that round. It is the central stochastic quantity in our model. The dynamics
of $X^{(j,n)}$ within an epoch form a recursion that we analyze in detail in the
subsequent sections.
\subsection{Block Reward}
A participant who finds the block in round $n$ of epoch $j$ at elapsed time $t$ collects two sources of revenue.
The first is a fixed block subsidy $R_b \ge 0$ issued by the protocol.
In Bitcoin, the subsidy changes only at a halving, which occurs every 210{,}000
blocks, roughly every four years \citep{bitcoinDevBlockChain}. One halving period
therefore contains about 104 epochs of $E=2016$ blocks. Our analysis concerns
individual rounds, epochs, and the long-run behavior of the network between two
halvings, and over this horizon the subsidy is constant. The effect of a different
subsidy level, such as the one after the next halving, is captured by comparing
different values of $R_b$. \citet{tsabary2018} and \citet{liReppenSircar2024} also
treat the subsidy as fixed.
The second is the transaction fee revenue collected from the transactions included in the block.

At elapsed time $t$, the available Type-1 supply is the carryover backlog $X$ from the previous round plus the new arrivals, given by $X + \lambda t$.
If this total Type-1 supply is smaller than the block capacity $s_b$, the block contains all of it.
The remaining $s_b - X - \lambda t$ units are then filled with Type-0 transactions.
If the Type-1 supply meets or exceeds $s_b$, the block is saturated entirely with Type-1 transactions.

The total reward for finding the block at elapsed time $t$ is therefore given by the function
\begin{equation}
\begin{aligned}
R(X, t) \;:=\;& R_b + \min\Bigl((X + \lambda t) r_1 \\
&+ (s_b - X - \lambda t) r_0,\;s_b\, r_1
\Bigr)
\end{aligned}
\label{eq:reward}
\end{equation}
This quantity is non-decreasing in $t$.
As more high-fee transactions arrive, the total reward grows until block saturation.
Before saturation, it grows linearly in $t$. \citet{kimSeo2026mind} report that an
affine reward function approximates observed Bitcoin fee data, and \citet{tsabary2018}
also let the fee reward grow with the time elapsed since the last block.
It reaches a maximum possible value $R_{\max} := R_b + s_b r_1$.
Further Type-1 arrivals beyond $s_b$ do not increase the fee revenue.
\subsection{Costs}
Each unit of hashpower $i\in I$ faces an operating cost rate $c_{\mathrm{op}}(i)>0$, charged
per unit of hashpower and per unit of time. This cost is paid only while $i$ is
active. A group $A\subseteq I$ active over an interval of length $t$ therefore pays
$t\int_A c_{\mathrm{op}}(i)\,d\eta(i)$ in total. As explained in \cref{sec:playerspace}, this
integral adds up the cost rates of all participants in $A$, each weighted by its
hashpower. $c_{\mathrm{op}}(\cdot)$ may vary arbitrarily across $I$. We place no
restriction on it beyond measurability and positivity for now, and postpone any
restriction on its cross-sectional shape to \cref{sec:aggregate}.

A separate capital cost rate $c_{\mathrm{cap}}(i)>0$ is paid per unit of hashpower and per unit
of time regardless of whether $i$ is active. Because it is paid either way, it plays
no role in the activation decision. This is the logic of a temporary shutdown, as in
models of mines and firms that operate only while their revenue covers their variable
cost \citep{brennan1985natural,mcdonald1985shutdown}. We set it aside until the free-entry analysis of
a later section.
\subsection{Strategies}
\label{sec:strategies}
At any elapsed time $t \ge 0$ during round $n$ of epoch $j$, each participant $i \in I$ decides
what fraction of their hashpower to supply to the network. Let $\mathbb{H}$ denote the space of
admissible public histories. We write $\mathcal{H}^{(j,n)}(t) \in \mathbb{H}$ for the realization
of the public history up to elapsed time $t$. Every participant observes $\mathcal{H}^{(j,n)}(t)$. No
participant sees anything beyond it.

A strategy for participant $i \in I$ is a non-anticipative mapping
\begin{equation}
a_i : \mathbb{H} \to [0,1],
\end{equation}
where $a_i\bigl(\mathcal{H}^{(j,n)}(t)\bigr)$ specifies the active fraction of participant $i$'s
hashpower at time $t$ under history $\mathcal{H}^{(j,n)}(t)$. We denote by $\mathcal{S}$ the common strategy
space consisting of all such mappings. A mapping is non-anticipative when the choice it
makes at time $t$ uses only the history up to $t$ and no information about the future.

A population strategy profile is an assignment $\mathbf{a} : I \to \mathcal{S}$ under which the active fraction $a_i(h)$ is measurable jointly in the participant $i$ and the history $h$, so that the aggregate in \eqref{eq:active_hashpower} below is well defined. It
assigns a strategy $a_i \in \mathcal{S}$ to each participant $i \in I$. To analyze individual
incentives, we denote the strategy profile of all participants other than $i$ as
$\mathbf{a}_{-i}$. Under profile $\mathbf{a}$, the aggregate active hashpower operating
on the network at history $\mathcal{H}^{(j,n)}(t)$ is given by the integral
\begin{equation}
\gamma_{\mathrm{act}}\!\left(\mathcal{H}^{(j,n)}(t)\right) \;:=\; \int_I a_i\left(\mathcal{H}^{(j,n)}(t)\right) d\eta(i)
\label{eq:active_hashpower}
\end{equation}

By definition, the active hashpower satisfies $\gamma_{\mathrm{act}} \le \gamma$, with equality holding if and only if every participant operates at full capacity.
\section{Miner Objective and the Myopic Activation Rule}
\label{sec:objective}
We now formalize the optimization problem for an individual participant.
At elapsed time $t$ in round $n$ of epoch $j$, the network operates with a total active hashpower $\gamma_{\mathrm{act}}(t)$ and a winning probability $w^{(j)}$.
Here and below we write $\gamma_{\mathrm{act}}(t)$ for $\gamma_{\mathrm{act}}\bigl(\mathcal{H}^{(j,n)}(t)\bigr)$, dropping the history argument for convenience whenever the round and epoch are clear from context.
The probability that no block has been found by time $t$ is given by the survival function
\begin{equation}
\bar{F}\bigl(t \mid X, \mathbf{a}_{-i}\bigr) \;:=\; \exp\left(-\int_0^t w^{(j)} \gamma_{\mathrm{act}}(u) du\right).
\end{equation}

If participant $i$ is active at elapsed time $t$, they discover blocks at rate $w^{(j)}$ and earn the reward $R(X, t)$, for a net payoff rate of $\pi_i(t) := w^{(j)} R(X, t) - c_{\mathrm{op}}(i)$.
The expected profit per unit of capacity over a single round is the integral of this payoff rate weighted by the survival probability.
We define this single-round expected profit as
\begin{equation}
\Pi_i^{(j,n)}(a_i ; \mathbf{a}_{-i}) \;:=\; \int_0^\infty a_i(t)\, \pi_i(t)\, \bar{F}\bigl(t \mid X, \mathbf{a}_{-i}\bigr)\, dt.
\end{equation}

Miners are economic agents who plan to maximize their profit across as many epochs as the network runs.
The rational objective for participant $i$ is therefore best represented as the long-run expected profit rate.
We define this objective directly over an infinite horizon, as the limiting ratio of cumulative expected profit to cumulative expected time. Writing $\bar\Pi_i^{(j)}(a_i ; \mathbf{a}_{-i}) := \sum_{n=1}^E \Pi_i^{(j,n)}(a_i ; \mathbf{a}_{-i})$ for the participant's total profit over epoch $j$, and recalling the epoch-average round length $\bar T^{(j)}$ from \eqref{eq:epochavg}, the objective is
\begin{equation}
\rho_i(a_i ; \mathbf{a}_{-i}) \;:=\; \liminf_{J \to \infty} \frac{\sum_{j=1}^J \mathbb{E}\bigl[\bar\Pi_i^{(j)}(a_i ; \mathbf{a}_{-i})\bigr]}{E\cdot\sum_{j=1}^J \mathbb{E}\bigl[\bar T^{(j)}\bigr]}.
\end{equation}

To solve this optimization problem, we establish the exogeneity of the network state.
Because the population is non-atomic, participant $i$ has a measure of exactly zero ($\eta(\{i\}) = 0$).
Any unilateral change to the individual strategy $a_i$ alters the population profile on a set of measure zero.
This implies exact invariance of the aggregate.
The active hashpower integral $\gamma_{\mathrm{act}}(t)$ is identical under any choice of $a_i$.

This exact invariance propagates through every level of the network dynamics.
First, the in-round survival function $\bar{F}$ remains unchanged.
Second, the probability distribution of the round duration $T^{(j,n)}$ is unchanged.
The carryover backlog $X$ for the subsequent round depends only on transaction arrivals and the realized duration of the previous round.
Therefore, the inter-round backlog trajectory is independent of participant $i$'s actions.
Finally, the expected durations of all rounds in an epoch are invariant.
The average epoch duration $\bar{T}^{(j)}$ is unchanged.
The difficulty-adjustment mechanism consequently yields the exact same winning probability $w^{(j+1)}$ for the next epoch.

The entire trajectory of public states $(X, w)$ and the expected time denominator in the long-run objective are exogenous to participant $i$.
The optimization of $\rho_i$ therefore decouples across rounds.
Participant $i$ maximizes the long-run ratio by maximizing the numerator $\Pi_i^{(j,n)}$ independently for each round.
For a non-atomic participant, maximizing the long-run objective is therefore equivalent to maximizing the expected profit of each round separately, which is the per-round objective of earlier gap games (\cref{sec:related}).
Our setting thus justifies that objective for individually negligible miners.
For a miner with a non-negligible share of the hashpower, whose actions move the backlog and the winning probability, the two objectives need not coincide.

The single-round profit $\Pi_i^{(j,n)}$ is a linear functional of the choice variable $a_i(t)$.
The survival function $\bar{F}$ acts as a strictly positive weight.
Participant $i$ maximizes the integral by setting $a_i(t) = 1$ whenever the instantaneous net profit rate is non-negative, and $a_i(t) = 0$ otherwise.
This pointwise maximization yields the optimal myopic activation rule:
\begin{equation}
a_i^*(t) \;=\;
\begin{cases}
1 & \text{if } w^{(j)} R(X, t) \ge c_{\mathrm{op}}(i) \\
0 & \text{otherwise.}
\end{cases}
\label{eq:myopic_rule}
\end{equation}

The reward function $R(X, t)$ is non-decreasing in elapsed time $t$.
The instantaneous net profit rate therefore crosses the zero threshold at most once per round.
This single crossing point ensures that the optimal strategy is a single continuous active phase.
The participant activates their hardware at the exact moment profitability is reached and remains active until the round ends.

The argument above holds for every profile $\mathbf{a}_{-i}$ of the other participants, so the myopic rule \eqref{eq:myopic_rule} is a weakly dominant strategy.
No other strategy gives participant $i$ a higher long-run profit rate, whatever strategies the other participants use.
The dominance is weak because a strategy that differs from \eqref{eq:myopic_rule} only on a set of times of zero length, such as the instant at which $w^{(j)} R(X, t) = c_{\mathrm{op}}(i)$, earns the same profit rate.
\begin{proposition}
\label{prop:myopic}
Fix any profile $\mathbf{a}_{-i}$ of the other participants under which rounds have a
finite expected length. Among all strategies of participant $i$, the myopic rule
\eqref{eq:myopic_rule} maximizes the long-run profit rate $\rho_i$, so it is weakly
dominant. Under this rule, participant $i$ is active from the first time at which
$w^{(j)}R(X,t)\ge c_{\mathrm{op}}(i)$ until the round ends.
\end{proposition}
\begin{remark}
The derived myopic rule places no restrictions on the profile $\mathbf{a}_{-i}$ of the other participants.
A non-atomic participant maximizes their long-run objective even when the rest of the network adopts arbitrary strategies.
Furthermore, this rule applies to heterogeneous operating costs without modification.
A larger $c_{\mathrm{op}}(i)$ implies a later activation threshold.
This outcome parallels the findings of finite-player mining gap models in the broader literature, which similarly conclude that higher operating costs delay activation \citep{kimSeo2026mind}.
The non-atomic framework, however, reaches this conclusion as an immediate consequence of exact invariance.
It avoids the complex equilibrium constructions required when participants hold a non-negligible share of the hashpower.
\end{remark}
\subsection{Operating Regime and Stability Restrictions}
\label{sec:regime}
To ensure the network functions continuously, we restrict our analysis to a sustainable operating regime.
This requires establishing structural boundaries on the reward mechanics and the transaction arrival rate.

The maximum possible reward a participant can extract from a single block is $R_{\max} := R_b + s_b r_1$.
We require the condition $w^{(j)} R_{\max} > c_{\mathrm{op}}(i)$ to hold for at least some subset of participants $A \subset I$ with positive measure $\eta(A) > 0$.
If this condition fails for the entire population, the net profit rate remains negative regardless of the accumulated transaction volume.
Participants remain idle.
This idleness triggers natural market and protocol corrections.
Users face prolonged wait times and increase their fee bids \citep{easley2019fees}.
This establishes a new two-tier fee structure with a higher maximum reward.
Simultaneously, the extended round duration decreases the epoch's average block rate.
The difficulty-adjustment mechanism then yields a higher winning probability $w^{(j+1)}$ for the subsequent epoch.
Both corrections naturally restore profitability.
We therefore model the network in its stabilized state where a positive measure of participants can already mine profitably.
\citet{kimSeo2026mind} adopt a similar assumption, arguing that if mining ceases for an extended period, the accumulating fees draw miners back.

A natural assumption is that the arrival rate of high-fee transactions does not overwhelm the network's clearing capacity.
If the expected number of Type-1 arrivals per round exceeds the block capacity $s_b$, the backlog grows indefinitely.
Type-1 bidders facing infinite wait times naturally escalate their bids to bypass the congestion.
This market reaction simply establishes a new two-tier fee structure at elevated price levels.
Our analysis therefore models the network after the market has settled into a sustainable fee equilibrium.
We analyze only Type-1 arrival rates that are small enough to maintain a stable backlog across rounds.
\citet{doger2024} make a similar assumption, taking the arrival rate of high-priority transactions small enough that their number in the mempool stays below the block capacity.
The transition periods in which fees and the winning probability are still adjusting are outside the scope of this paper and are left for future research.
\section{Aggregate Behavior Under a Common Cost}
\label{sec:aggregate}
The myopic rule \eqref{eq:myopic_rule} holds for an arbitrary measurable cost
function $c_{\mathrm{op}}:I\to(0,\infty)$, but the resulting aggregate
$\int_I a_i^*(\cdot)\,d\eta(i)$ then depends on the full cross-sectional
distribution of costs across $I$. To obtain a tractable
aggregate, we restrict attention to the case where operating costs are homogeneous
across the population. This restriction is a common baseline in the mining gap and blockchain mining literature
\citep{tsabary2018,liReppenSircar2024,noguchi2025,kimSeo2026mind,budish2025trust}.
\citet{noguchi2025}, for example, assumes identical costs across miners and notes that
this is done for mathematical tractability.
When every participant faces the same threshold, the population acts as a
single unit and no participant activates their rig earlier than any other.
Heterogeneity in costs would spread activation times
across a range of backlog levels, and participants with lower costs would begin
mining before those with higher costs. The resulting aggregate would no longer be a
single threshold, and the backlog analysis of \cref{sec:backlog} would need a
different approach. We leave heterogeneous costs for future research.

The myopic rule \eqref{eq:myopic_rule} was derived without any restriction on the
cross-sectional shape of $c_{\mathrm{op}}$. We impose the restriction below for the aggregate
analysis of this section and all subsequent sections.
\begin{assumption}[Common Operating Cost]
\label{ass:commoncost}
There exists $\bar c_{\mathrm{op}}>0$ such that $c_{\mathrm{op}}(i)=\bar c_{\mathrm{op}}$ for
$\eta$-almost every $i\in I$.
\end{assumption}

The phrase ``$\eta$-almost every $i$'' means every participant except possibly a group
that controls zero hashpower in total. Such a group contributes nothing to any
integral over $I$, so it cannot affect the active hashpower or any other aggregate
quantity. The assumption therefore amounts to a common operating cost for the network
as a whole.

Under \cref{ass:commoncost}, the myopic rule activates participant $i$ at
elapsed time $t$ if and only if $w\,R(X,t)\ge\bar c_{\mathrm{op}}$. Because $R(X,t)$
is non-decreasing in the total Type-1 supply $X+\lambda t$, this condition is
met once $X+\lambda t$ reaches a level that depends only on $w$. We call this
level the \emph{mining threshold}:
\begin{equation}
s_1(w) \;:=\; \max\!\left(
\frac{\bar c_{\mathrm{op}}/w \;-\; R_b \;-\; s_b\,r_0}{r_1-r_0},\;0
\right).
\label{eq:s1common}
\end{equation}

The expression inside the max is the total Type-1 supply at which the reward
rate $w\,R(X,t)$ first covers $\bar c_{\mathrm{op}}$. It is obtained by solving
$w\bigl[R_b + s_b\,r_0 + (X+\lambda t)(r_1-r_0)\bigr] = \bar c_{\mathrm{op}}$
for $X+\lambda t$ in the non-saturated regime. The max with zero captures the
case where even an empty mempool already provides sufficient reward. In that case,
$s_1(w)=0$ and no gap forms.

Under the myopic rule, each participant's activation depends on the history
$\mathcal{H}^{(j,n)}(t)$ only through the current total Type-1 supply
$x = X+\lambda t$. Recall that \cref{sec:objective} already writes $\gamma_{\mathrm{act}}(t)$ for $\gamma_{\mathrm{act}}\bigl(\mathcal{H}^{(j,n)}(t)\bigr)$. Here we overload the same shorthand once more and write $\gamma_{\mathrm{act}}(x)$ in place of
$\gamma_{\mathrm{act}}(\mathcal{H}^{(j,n)}(t))$ from
\eqref{eq:active_hashpower}. Every participant outside a null set shares the
threshold $s_1(w)$, so the aggregate active hashpower at total Type-1 supply
$x$ is
\begin{equation}
\begin{aligned}
\gamma_{\mathrm{act}}(x)
\;&=\;\int_I \mathbf{1}\{x\ge s_1(w)\}\,d\eta(i)\\
\;&=\;\gamma\cdot\mathbf{1}\{x\ge s_1(w)\}
\label{eq:aggregatestep}
\end{aligned}
\end{equation}

The population is either fully active or fully idle. The null set on which
costs may differ contributes nothing to the integral, since $\eta$ assigns it
zero mass.
\section{Backlog Dynamics}
\label{sec:backlog}
The winning probability $w$ is constant within a single epoch and changes only at
epoch boundaries through the difficulty-adjustment mechanism
\eqref{eq:daaprotocol}. For the analysis of this section we fix $w$ and
study the backlog recursion within one epoch. The evolution of $w$ across
epochs is deferred to \cref{sec:daa}. This incurs no loss of generality.
Every quantity derived below is a function of $w$, and substituting
$w = w^{(j)}$ recovers the epoch-specific behavior.
The mining threshold $s_1(w)$ from \eqref{eq:s1common} is one such quantity, and it is constant throughout the epoch because $w$ is. For the remainder of this section we therefore drop the argument and write $s_1$ for $s_1(w)$, reintroducing the explicit dependence on $w$ once \cref{sec:daa} lets $w$ vary again.
For convenience, we write $\mu := w\gamma$ for the aggregate block-finding rate when the full
population is active.
\subsection{The Recursion}
At the start of round $n$ the mempool holds backlog $X^{(n)}$. If
$X^{(n)} \geq s_1$, the full population is active immediately. If
$X^{(n)} < s_1$, the population remains idle until Type-1 arrivals,
which occur at the constant rate $\lambda$, bring the backlog up to $s_1$. In either case, mining begins with
$\max(X^{(n)},\,s_1)$ Type-1 transactions present. Once mining begins, the population finds a block after a random duration $T_\mu^{(n)}$, the time until the first success at the aggregate rate $\mu$, so $T_\mu^{(n)} \sim \mathrm{Exp}(\mu)$. During this active
phase a further
$\lambda T_\mu^{(n)}$ Type-1 transactions arrive. The block then clears up
to $s_b$ of them, and the remainder carries into round $n+1$:
\begin{equation}
X^{(n+1)} = \Bigl(\max\bigl(X^{(n)},\,s_1\bigr) + \lambda\,T_\mu^{(n)}
- s_b\Bigr)^{+}.
\label{eq:recursion}
\end{equation}

This is a Lindley-type recursion with a reflecting barrier at zero and a
floor at $s_1$. The floor arises because mining never begins below the
threshold.
\subsection{Stationary Backlog Distribution}
Each epoch contains $E$ rounds. In Bitcoin, $E = 2016$. Under the
stability condition $\lambda/\mu < s_b$ discussed in \cref{sec:regime}, the recursion \eqref{eq:recursion} is ergodic. That is,
it admits a unique stationary distribution, the long-run distribution that the
backlog settles into, and converges to it from every initial condition. We treat the within-epoch distribution as stationary.
This is an approximation. The first few rounds of each epoch carry a
transient component inherited from the previous epoch's terminal state.
We neglect this component.

Write $\theta := \mu/\lambda$, $p_0 := \Pr(X = 0)$ for the probability that the
stationary backlog $X$ is zero, and $f_X$ for its density on $(0,\infty)$. The
stationary distribution depends on $\theta$ through the characteristic equation
\begin{equation}
K = \theta\bigl(1 - e^{-Ks_b}\bigr).
\label{eq:char}
\end{equation}
Under the stability condition, this equation has exactly one positive root
(\cref{lem:root}), which we denote by $K$. It has the closed form $K = \theta + W_0\bigl(-\theta s_b e^{-\theta s_b}\bigr)/s_b$, where
$W_0$ is the principal branch of the Lambert W function, the inverse of $z\mapsto ze^{z}$
on $[-1,\infty)$ \citep{corless1996lambert}.

The stationary distribution of the backlog is then
fully specified:
\begin{equation}
\begin{aligned}
p_0 &= 1 - e^{-K(s_b - s_1)}, \\\qquad
f_X(x) &= K\,e^{-K(s_b-s_1)}\,e^{-Kx}, \quad x > 0
\label{eq:stationarybacklog}
\end{aligned}
\end{equation}
\begin{proposition}
\label{prop:backlog}
Suppose $\lambda/\mu<s_b$ and $s_1<s_b$. The recursion \eqref{eq:recursion} has a unique
stationary distribution, and the backlog converges to it in distribution from every
initial condition. Under this distribution, $\Pr(X>x)=e^{-K(s_b-s_1+x)}$ for $x\ge0$, so
$X$ has the atom $p_0$ at zero and the density $f_X$ of \eqref{eq:stationarybacklog},
where $K$ is the unique positive root of \eqref{eq:char}.
\end{proposition}
\Cref{ec:backlog} derives the transient round-by-round distribution of the
backlog, obtains the stationary distribution from it, and proves
\cref{prop:backlog}.

The distribution is a mixture. The point mass $p_0$ is the probability that
the backlog is exactly zero at the start of a round. The exponential density
governs rounds where some Type-1 transactions carry over. A larger $K$ means
the density decays faster, so large backlogs are less likely.
\Cref{fig:mempool} illustrates a sample path of the backlog dynamics, and
\cref{fig:dists} displays the resulting stationary distributions.
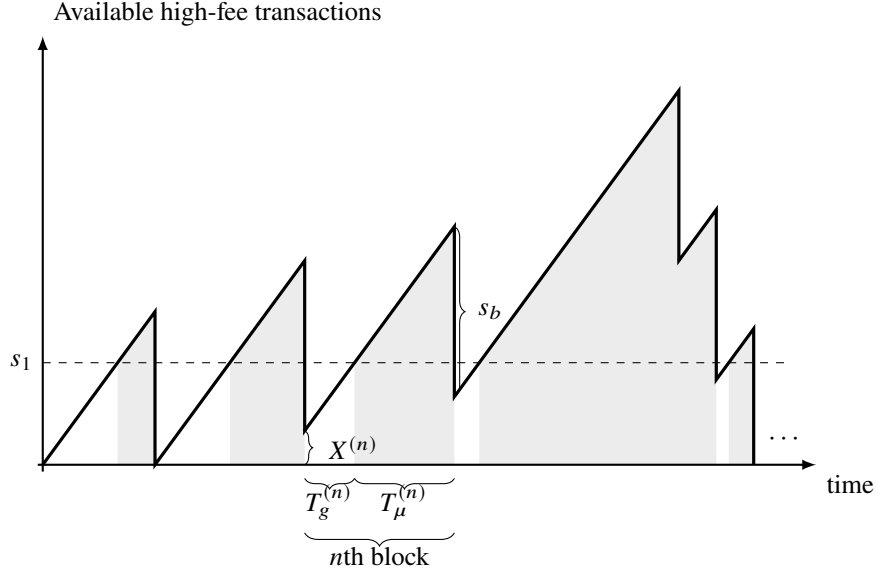
\begin{figure}[tb]
\centering
\begin{tikzpicture}[x=0.33cm, y=0.45cm, >=latex, every node/.style={font=\small}]
\fill[gray!15] (3,0)    -- (3,3)    -- (4.5,4.5) -- (4.5,0)  -- cycle;  %
\fill[gray!15] (7.5,0)  -- (7.5,3)  -- (10.5,6)  -- (10.5,0) -- cycle;  %
\fill[gray!15] (12.5,0) -- (12.5,3) -- (16.5,7)  -- (16.5,0) -- cycle;  %
\fill[gray!15] (17.5,0) -- (17.5,3) -- (25.5,11) -- (25.5,0) -- cycle;  %
\fill[gray!15] (25.5,0) -- (25.5,6) -- (27,7.5)  -- (27,0)   -- cycle;  %
\fill[gray!15] (27.5,0) -- (27.5,3) -- (28.5,4)  -- (28.5,0) -- cycle;  %
\draw[dashed] (0,3) -- (30,3); \node[left] at (0,3) {$s_1$};
\draw[->,thick] (-0.2,0) -- (31,0) node[below right] {time};
\draw[->,thick] (0,-0.2) -- (0,12.6) node[above right] {Available high-fee transactions};
\draw[very thick]
(0,0) -- (4.5,4.5) -- (4.5,0) -- (10.5,6) -- (10.5,1)
-- (16.5,7) -- (16.5,2) -- (25.5,11) -- (25.5,6)
-- (27,7.5) -- (27,2.5) -- (28.5,4) -- (28.5,0);
\node at (29.8,0.7) {$\cdots$};
\draw[decorate,decoration={brace,amplitude=4pt}] (16.5,7) -- (16.5,2);
\node[right=5pt] at (16.5,4.5) {$s_b$};
\draw[decorate,decoration={brace,amplitude=3pt}] (10.5,1) -- (10.5,0);
\node[right=4pt] at (10.6,0.5) {$X^{(n)}$};
\draw[decorate,decoration={brace,amplitude=4pt}] (12.5,-0.4) -- (10.5,-0.4);
\node at (11.5,-1.1) {$T_g^{(n)}$};
\draw[decorate,decoration={brace,amplitude=4pt}] (16.5,-0.4) -- (12.5,-0.4);
\node at (14.5,-1.1) {$T_\mu^{(n)}$};
\draw[decorate,decoration={brace,amplitude=5pt}] (16.5,-2.0) -- (10.5,-2.0);
\node at (13.5,-2.7) {$n$th block};
\end{tikzpicture}
\caption{A sample path of the available high-fee transaction count across several
rounds. The backlog rises at the arrival rate $\lambda$. The miners are idle whenever the backlog stays below the mining threshold $s_1$. This creates the unshaded gap $T_g^{(n)}$.
Once the backlog reaches $s_1$, they mine until a block is found. This corresponds to the shaded active phase
$T_\mu^{(n)}$. At block discovery, up to $s_b$ transactions are cleared, which creates the
vertical drop, and the leftover $X^{(n)}$ is carried into the next round. A
carryover above $s_1$ lets mining begin with no gap.}
\label{fig:mempool}
\end{figure}
\begin{figure}[tb]
\centering
\begin{subfigure}[t]{0.48\textwidth}
\centering
\begin{tikzpicture}[scale=0.8, >=latex, every node/.style={font=\small}]
\draw[->,thick] (-0.2,0) -- (6.2,0) node[below right] {$x$};
\draw[->,thick] (0,-0.2) -- (0,4.4) node[above left] {$f_X(x)$};
\draw[very thick,smooth,domain=0:9,samples=40,variable=\x]
plot ({0.6*\x},{2.4*exp(-0.8-0.4*\x)});
\draw[very thick] (0,0) -- (0,3.3);
\fill (0,3.3) circle (2.2pt);
\node[right] at (0.15,3.45) {$p_0$};
\end{tikzpicture}
\caption{Backlog $X$: an atom $p_0=1-e^{-K(s_b-s_1)}$ at the origin and the
exponential density $Ke^{-K(s_b-s_1)}e^{-Kx}$ for $x>0$.}
\label{fig:dist-X}
\end{subfigure}
\hfill
\begin{subfigure}[t]{0.48\textwidth}
\centering
\begin{tikzpicture}[scale=0.8, >=latex, every node/.style={font=\small}]
\draw[->,thick] (-0.2,0) -- (6.2,0) node[below right] {$t$};
\draw[->,thick] (0,-0.2) -- (0,4.4) node[above left] {$f_{T_g}(t)$};
\draw[very thick,smooth,domain=0:3,samples=30,variable=\x]
plot ({1.6*\x},{2.4*exp(-2+0.4*\x)});
\draw[very thick] (0,0) -- (0,0.81);
\fill (0,0.81) circle (2.2pt);
\node[right] at (0.15,1.0) {$e^{-Ks_b}$};
\draw[very thick] (4.8,0) -- (4.8,3.3);
\fill (4.8,3.3) circle (2.2pt);
\node[right] at (4.95,3.45) {$p_0$};
\draw (4.8,0.12) -- (4.8,-0.12);
\node[below] at (4.8,-0.18) {$s_1/\lambda$};
\end{tikzpicture}
\caption{Mining gap $T_g$: atoms $e^{-Ks_b}$ at $0$ and $p_0$ at the maximal gap
$s_1/\lambda$, with the increasing density $K\lambda e^{-K(s_b-\lambda t)}$ between.}
\label{fig:dist-Tg}
\end{subfigure}
\caption{Stationary distributions of the Type-1 backlog $X$ and the mining gap
$T_g$. Both are mixed distributions, drawn here with the probability
masses, shown as dots and the continuous densities on the same axes. The empty
mempool $\{X=0\}$ carries probability $p_0$. The same mass $p_0$ reappears as the atom at
the maximal gap in the mining gap distribution $T_g$. $\{X\ge s_1\}$ becomes the no-gap
atom $e^{-Ks_b}$ at the origin for the mining gap.}
\label{fig:dists}
\end{figure}
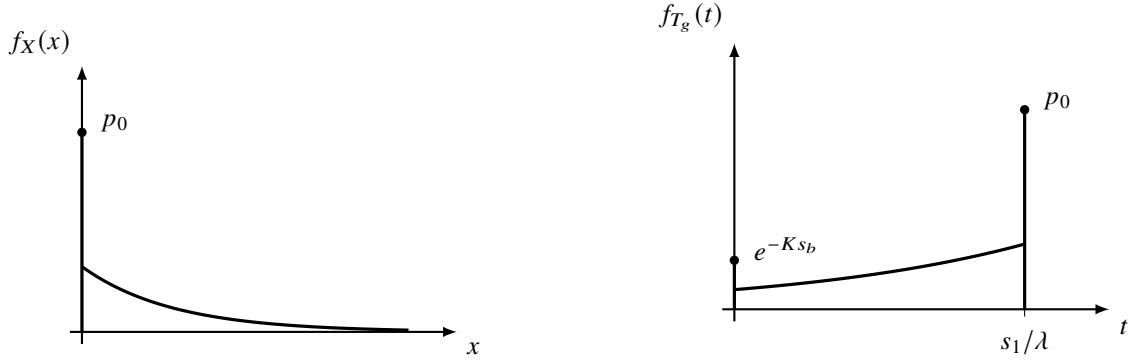
\subsection{Mining Gap Distribution and Expectation}
The stationary approximation adopted above applies equally to the gap
distribution derived from it. The gap is
\begin{equation}
T_g \;:=\; \frac{\bigl(s_1 - X\bigr)^{+}}{\lambda}.
\label{eq:gapdef}
\end{equation}

The gap is zero whenever $X \geq s_1$. It equals
$(s_1 - X)/\lambda$ whenever $X < s_1$. The maximum possible gap is
$s_1/\lambda$, reached when the mempool is empty.

The stationary backlog distribution \eqref{eq:stationarybacklog} gives the
distribution function of the gap,
\begin{equation}
F_{T_g}(t) =
\begin{cases}
0 & t < 0,\\[4pt]
e^{-K(s_b - \lambda t)} & 0 \leq t < s_1/\lambda,\\[4pt]
1 & t \geq s_1/\lambda.
\end{cases}
\label{eq:gapcdf}
\end{equation}

The distribution has atoms at zero and at the maximum $s_1/\lambda$.
The atom at zero is $\Pr(X \geq s_1) = e^{-Ks_b}$, the probability that
the carryover backlog already meets the threshold. The atom at
$s_1/\lambda$ is $p_0$, the probability that the mempool is empty.

The expected gap is
\begin{equation}
\mathbb{E}[T_g] = \frac{s_1}{\lambda}
- \frac{e^{-Ks_b}}{K\lambda}\bigl(e^{Ks_1} - 1\bigr).
\label{eq:expectedgap}
\end{equation}
\begin{corollary}
\label{cor:gap}
Under the conditions of \cref{prop:backlog} and $s_1>0$, the stationary mining gap
$T_g$ has the distribution function \eqref{eq:gapcdf}, with the atom $e^{-Ks_b}$ at zero
and the atom $p_0$ at $s_1/\lambda$, and the mean \eqref{eq:expectedgap}.
\end{corollary}

The first term is the maximum possible gap, corresponding to an empty
mempool. The second term subtracts the expected time saved across rounds
where the carryover backlog is large enough to shorten or eliminate the
waiting period. In a mempool that resets after every block
\citep{tsabary2018,noguchi2025,kimSeo2026mind}, every round starts from an empty mempool
and the gap equals $s_1/\lambda$ in every round. The carryover makes the gap random, with
an atom at zero, and the second term of \eqref{eq:expectedgap} shortens its mean. This
term vanishes as $K$ grows, so at high aggregate block-finding rates the backlog behaves
as in a mempool that resets after every block.
\section{The Difficulty Adjustment Map and Stability}
\label{sec:daa}
\subsection{The Mean Round Length}
The expected round length at winning probability $w$ is the sum of the expected
gap and the expected active phase. Once the gap ends, every participant is active, and the population as a whole finds a block at the aggregate rate $\mu$, so the active phase has expected duration $1/\mu$. From \eqref{eq:expectedgap} and
$\mu = w\gamma$, the expected round length is
\begin{equation}
m(w) \;:=\; \mathbb{E}[T_g] + \frac{1}{\mu}.
\label{eq:mdef}
\end{equation}

Combining \eqref{eq:expectedgap} with the characteristic equation
\eqref{eq:char} simplifies this to
\begin{equation}
m(w) \;=\; \frac{s_1(w)}{\lambda}
+ \frac{1 - e^{-K(w)\left(s_b - s_1(w)\right)}}{K(w)\,\lambda},
\label{eq:mw}
\end{equation}
where $K(w)$ is the unique positive root
of \eqref{eq:char} evaluated at $\mu = w\gamma$. The function $m(w)$ is strictly decreasing
in $w$, since a higher winning probability lowers the mining threshold $s_1(w)$ and raises
the aggregate block-finding rate $\mu$, and both effects shorten the round (\cref{app:stability}).
\subsection{The Protocol Map}
The difficulty-adjustment mechanism \eqref{eq:daaprotocol} maps the
current winning probability to the next one through the realized average round
length. Within an epoch, $w$ is fixed, and the rounds are stationary and
ergodic. Since $E$ is large, the ergodic theorem gives
$\bar{T}^{(j)} \approx m(w^{(j)})$. The protocol map in the unclamped
region is therefore
\begin{equation}
\Phi(w) \;:=\; w\cdot\frac{m(w)}{B}.
\label{eq:phimap}
\end{equation}

The clamp in \eqref{eq:daaprotocol} restricts any single adjustment to a
fourfold change. We work throughout in the unclamped region, where
$\Phi(w) = w\,m(w)/B$. The clamp binds only when the network is far from
its invariant point, and the local analysis below is unaffected by it.
\citet{kimSeo2026mind} drop the clamp in the same way, arguing that adjustments
beyond this range are highly unlikely.

If blocks came too fast ($m(w) < B$), then $\Phi(w) < w$, and the protocol
reduces the winning probability, making the puzzle harder. If blocks came too
slow ($m(w) > B$), then $\Phi(w) > w$, and the protocol increases the
winning probability. The map $\Phi$ encodes this correction.
\subsection{The Invariant Winning Probability}
A fixed point of \eqref{eq:phimap} satisfies $\Phi(w^*) = w^*$, which
requires $m(w^*) = B$. This is the only candidate for a steady state.
If $m(w) \neq B$, then $\Phi(w) \neq w$, and the protocol keeps
adjusting.

Existence follows from the intermediate value theorem. As $w \to \infty$,
the threshold $s_1(w) = 0$ and $\mu = w\gamma \to \infty$, so
$m(w) \to 0$. As $w$ decreases toward the minimum viable level
$\underline{w} := \bar{c}_{\mathrm{op}}/R_{\max}$, mining stops being worthwhile at all. This is exactly the point below which the profitability condition of \cref{sec:regime} fails for the entire population. At $\underline w$, the threshold $s_1(w) \to s_b$,
so $s_b - s_1(w) \to 0$, and $m(w) \to s_b/\lambda$. The condition
$s_b/\lambda > B$ says that filling a block through Type-1 arrivals alone,
without any mining, would take longer than the protocol's target. Under this condition,
$m$ exceeds $B$ near $\underline{w}$ and falls below $B$ for large $w$.
By continuity, there exists $w^*$ with $m(w^*) = B$.

Since $m(w)$ is strictly decreasing as $w$ increases, it crosses the level $B$ exactly once, and the invariant winning probability $w^*$ is unique.
\begin{proposition}
\label{prop:invariant}
Suppose $\lambda B<s_b$, and let $w$ range over the winning probabilities above
$\max\bigl(\underline w,\lambda/(\gamma s_b)\bigr)$, at which some mining is profitable
and the backlog is stable. On this range, the mean round length $m(w)$ in \eqref{eq:mw} is
continuous and strictly decreasing. It tends to $s_b/\lambda$ at the lower end and to zero
as $w\to\infty$, so there is a unique invariant winning probability $w^*$ with
$m(w^*)=B$.
\end{proposition}
\subsection{Local Stability}
The question is the following. Suppose the system is at $w^*$ and a
perturbation moves expected block times such that it is as if the winning probability were $w^* + \delta$ for small $\delta$. After one
epoch, the map sends it to $\Phi(w^* + \delta)$. Does the perturbation
shrink or grow?

A first-order expansion gives
$\Phi(w^* + \delta) \approx w^* + \Phi'(w^*)\delta$. The perturbation
in the next epoch is $\Phi'(w^*)\delta$. After $j$ epochs it is
$[\Phi'(w^*)]^j \delta$. The fixed point attracts nearby trajectories if
$|\Phi'(w^*)| < 1$ and repels them if $|\Phi'(w^*)| > 1$.

Differentiating $\Phi(w) = w\,m(w)/B$ and using $m(w^*) = B$ gives
\begin{equation}
\Phi'(w^*) = 1 + \frac{w^*\,m'(w^*)}{m(w^*)} = 1 + \varepsilon(w^*),
\label{eq:phiprime}
\end{equation}
where
\begin{equation}
\varepsilon(w) := \frac{w\,m'(w)}{m(w)}
\label{eq:elasticity}
\end{equation}
is the elasticity of the mean round length with respect to the winning
probability. It measures the percentage change in round length per percentage
change in $w$.

The stability criterion $|\Phi'(w^*)| < 1$ becomes
\begin{equation}
-2 < \varepsilon(w^*) < 0.
\label{eq:stabilitycriterion}
\end{equation}

This is the same criterion identified by \citet{kimSeo2026mind} for a
population of finitely many miners under a deterministic mempool, restated
here for the mean round length $m(w)$ induced by the stochastic backlog of
\cref{sec:backlog}.

Since $m$ is strictly decreasing, $m'(w) < 0$, so $\varepsilon(w) < 0$
for all $w$. The right inequality in \eqref{eq:stabilitycriterion} holds
automatically. The substantive condition is $\varepsilon(w^*) > -2$.
\paragraph{The no-gap case.}
When $s_1(w^*) = 0$, the mean round length reduces to $m(w) = 1/(w\gamma)$,
and its elasticity is $\varepsilon(w) = -1$. This lies in $(-2, 0)$ unconditionally, so the no-gap equilibrium is always
stable. An elasticity of exactly $-1$ has a clean reading. By
\eqref{eq:phiprime}, a perturbation $\delta$ maps to $\Phi'(w^*)\delta =
(1+\varepsilon(w^*))\delta = 0$, so it is entirely undone after a single
epoch, and the winning probability lands back on $w^*$ exactly, as if the
perturbation had never occurred. This full cancellation in a single epoch is the extreme case. More generally, whenever $|\Phi'(w^*)|<1$ but $\Phi'(w^*)\neq 0$, perturbations decay geometrically over several epochs, as illustrated in \cref{fig:stability}.
\paragraph{The gap case.}
When $s_1(w^*) > 0$, whether $\varepsilon(w^*)$
remains above $-2$ depends on the equilibrium values $(K^*, s_1^*, \mu^*)$
through an explicit closed-form condition, which we derive in full in \cref{app:stability} for
the interested reader. As the aggregate hashpower grows, $\mu^*$ grows without bound
(\cref{sec:hashpower}), so $K^*$ grows as well and the probability $p_0$ that a round
starts with an empty mempool tends to one. The mempool then behaves as one that resets
after every block, and $\varepsilon(w^*)$ approaches
$-1-(R_b+s_b r_0)/\bigl((r_1-r_0)\lambda B\bigr)$, which is the elasticity that
\citet{kimSeo2026mind} obtain for such a mempool (\cref{app:stability}). For large
aggregate hashpower, the invariant winning probability is therefore stable when the fee
premium that Type-1 arrivals add over a target round length, $(r_1-r_0)\lambda B$,
exceeds the baseline block revenue $R_b+s_b r_0$, and unstable when it falls short.
As with the other structural assumptions in this
paper, we proceed by working in the regime where this condition holds,
taking the best case for the sake of the argument. When
\eqref{eq:stabilitycriterion} fails, the fixed point repels nearby
trajectories. \citet{kimSeo2026mind} show, for a fixed non-stochastic
mempool, that such an unstable invariant point induces winning probabilities that
oscillate between two values rather than converging.
\Cref{fig:oscillation} illustrates the qualitative picture.
\begin{proposition}
\label{prop:stability}
Suppose $\lambda B<s_b$ and that $w^*$ differs from
$\bar c_{\mathrm{op}}/(R_b+s_b r_0)$, the winning probability at which the mining
threshold \eqref{eq:s1common} reaches zero. Then $w^*$ is locally asymptotically stable
under the map \eqref{eq:phimap} if $\varepsilon(w^*)>-2$ and unstable if
$\varepsilon(w^*)<-2$. Without a gap, $\varepsilon(w^*)=-1$. With a gap,
$\varepsilon(w^*)$ has the closed form given in \cref{app:stability}, and as
$\gamma\to\infty$ it converges to $-1-(R_b+s_b r_0)/\bigl((r_1-r_0)\lambda B\bigr)$.
\end{proposition}
\begin{figure}[t]
\centering
\begin{subfigure}[t]{0.48\textwidth}
\centering
\begin{tikzpicture}[scale=0.75,>=latex,every node/.style={font=\small}]
\draw[->,thick] (-0.3,0) -- (6.3,0) node[below right] {$w^{(j)}$};
\draw[->,thick] (0,-0.3) -- (0,6.3) node[above left] {$w^{(j+1)}$};
\draw[dashed] (0,0) -- (6,6);
\node[below right] at (6,6) {\small $45^\circ$};
\draw[very thick,domain=0.4:5.8,samples=30,smooth,variable=\x]
plot ({\x},{3+0.4*(3-\x)/(0.6+abs(3-\x))});
\node[above] at (5.6,1.6) {$\Phi(w)$};
\fill (3,3) circle (2pt);
\node[below right] at (3,3) {$w^*$};
\draw[->,blue] (1.8,0) -- (1.8,3.27) -- (3.27,3.27) -- (3.27,2.88)
  -- (2.88,2.88) -- (2.88,3.07) -- (3.07,3.07) -- (3.07,2.96);
\node[blue,left] at (1.7,0.5) {convergence};
\end{tikzpicture}
\caption{Stable fixed point, $|\Phi'(w^*)|<1$. A small deviation from $w^*$ becomes
smaller from one epoch to the next, so the iterates move toward the fixed point. In the
no-gap special case, $\Phi'(w^*)=0$, and the iteration reaches $w^*$ in a single epoch.}
\label{fig:stability}
\end{subfigure}
\hfill
\begin{subfigure}[t]{0.48\textwidth}
\centering
\begin{tikzpicture}[scale=0.75,>=latex,every node/.style={font=\small}]
\draw[->,thick] (-0.3,0) -- (6.3,0) node[below right] {$w^{(j)}$};
\draw[->,thick] (0,-0.3) -- (0,6.3) node[above left] {$w^{(j+1)}$};
\draw[dashed] (0,0) -- (6,6); \node[below right] at (6,6) {\small $45^\circ$};
\draw[very thick,domain=0.4:5.8,samples=30,smooth,variable=\x]
plot ({\x},{3+2.6*(3-\x)/(0.6+abs(3-\x))});
\node[above] at (5.6,1.3) {$\Phi(w)$};
\fill (3,3) circle (2pt);
\node[below right] at (3,3) {$w^*$};
\draw[->,red] (2.4,0) -- (2.4,4.4) -- (4.4,4.4) -- (4.4,1.2)
-- (1.2,1.2) -- (1.2,5.4);
\node[red,left] at (1.1,5.4) {oscillation};
\end{tikzpicture}
\caption{Unstable fixed point, $|\Phi'(w^*)|>1$. A small deviation from $w^*$ is
amplified from one epoch to the next, causing the iterates to move away from the fixed
point and oscillate rather than converge.}
\label{fig:oscillation}
\end{subfigure}
\caption{Illustration of the iteration $\Phi$ near the invariant winning probability
$w^*$. A perturbation $\delta$ at epoch $j$ becomes approximately $\Phi'(w^*)\delta$ at
epoch $j+1$.}
\label{fig:cobweb}
\end{figure}
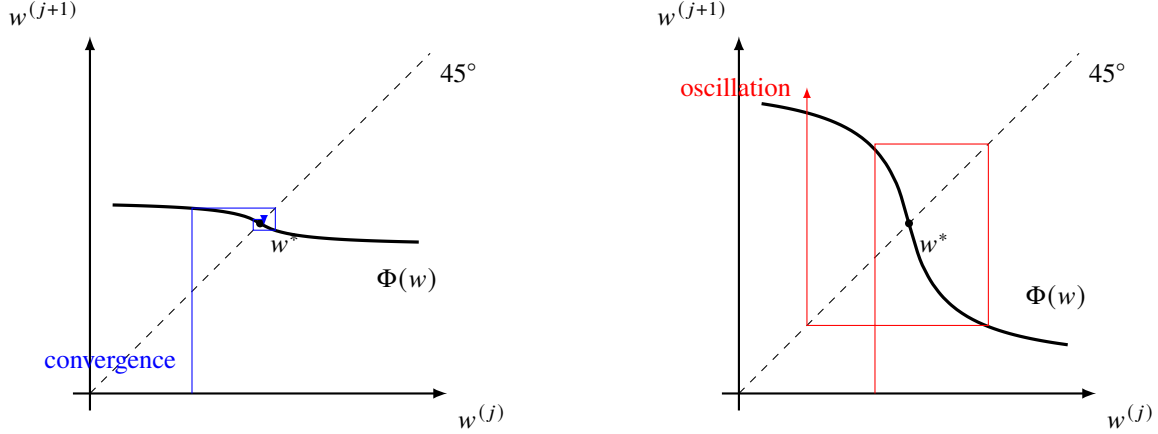

\section{Security Metrics}
\label{sec:metrics}
Three quantities jointly characterize the security state of the network, each capturing a
different dimension of vulnerability.
\subsection{Mining Gap}
The mining gap $T_g$ is the idle period at the start of a round during which all honest miners
are idle. This is a direct security concern. During the gap, an adversarial miner who continues
hashing uncontested can extend a private chain faster than the honest network. Nakamoto's
classical double-spending analysis assumes continuous honest participation, and the mining
gap violates this assumption for a fraction approximately $\mathbb{E}[T_g]/B$ of every round. A longer expected gap
therefore corresponds to a weaker security guarantee, increasing the success probability of
adversaries with less than half of the total hashpower \citep{nakamoto2008}. We use the expected gap length
$\mathbb{E}[T_g]$ as our measure of mining gap vulnerability.
\subsection{Aggregate Hashpower}
The aggregate hashpower $\gamma$ is the total hash attempts per second made by the honest network while its miners are active. A higher aggregate hashpower
raises the absolute hashpower behind any given fraction of the network. A double-spending or selfish mining attack succeeds more often as the attacker's share of hashpower grows. A higher aggregate hashpower makes a given share more costly to acquire, in hardware and energy. We use aggregate hashpower as our measure of the network's resistance to a hashpower-based attack.
\subsection{Hashwork}
The \emph{hashwork} of a round, a term we introduce here, is the total number of hash attempts
made by all honest miners before a valid block is found. Once mining begins, the population as a whole finds a block at the aggregate rate $\mu = w\gamma$, so the expected duration of the active mining phase is $1/\mu$. With aggregate hashpower $\gamma$ each hashing for an expected time of $1/\mu$, the expected total number of hash attempts, i.e. expected hashwork, in a
round is
\begin{equation}
\mathcal{W} \;:=\; \frac{\gamma}{\mu} \;=\; \frac{1}{w}.
\label{eq:hashwork_def}
\end{equation}

A higher expected hashwork means an adversary must expend more
computational resources to produce a competing chain, so it measures the raw cost of a
brute-force attack. We use the expected hashwork $\mathcal{W}$ as our measure of
attack-cost hardness. \Cref{fig:hashwork} illustrates this notion and contrasts it with the gap.
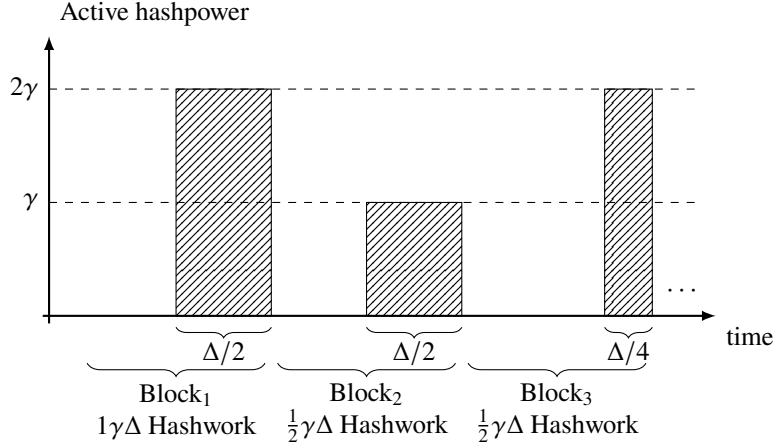
\begin{figure}[tb]
\centering
\begin{tikzpicture}[x=0.42cm, y=0.5cm, >=latex, every node/.style={font=\small}]
\draw[dashed] (0,3) -- (20.5,3);
\draw[dashed] (0,6) -- (20.5,6);

\draw[pattern=north east lines] (4,0)    rectangle (7,6);
\draw[pattern=north east lines] (10,0)   rectangle (13,3);
\draw[pattern=north east lines] (17.5,0) rectangle (19,6);

\draw[->,thick] (-0.2,0) -- (21,0) node[below right] {time};
\draw[->,thick] (0,-0.2) -- (0,7.4) node[above right] {Active hashpower};
\node[left] at (0,3) {$\gamma$};
\node[left] at (0,6) {$2\gamma$};
\node at (20,0.6) {$\cdots$};

\draw[decorate,decoration={brace,amplitude=4pt}] (7,-0.3) -- (4,-0.3);
\node at (5.5,-1.0) {$\Delta/2$};

\draw[decorate,decoration={brace,amplitude=4pt}] (13,-0.3) -- (10,-0.3);
\node at (11.5,-1.0) {$\Delta/2$};

\draw[decorate,decoration={brace,amplitude=4pt}] (19,-0.3) -- (17.5,-0.3);
\node at (18.25,-1.0) {$\Delta/4$};

\draw[decorate,decoration={brace,amplitude=5pt}] (6.8,-1.25) -- (1.2,-1.25);
\node[align=center] at (4,-2.55) {Block$_1$\\ $1 \gamma\Delta $ Hashwork};

\draw[decorate,decoration={brace,amplitude=5pt}] (12.8,-1.25) -- (7.2,-1.25);
\node[align=center] at (10,-2.55) {Block$_2$\\ $\tfrac12 \gamma\Delta$ Hashwork};

\draw[decorate,decoration={brace,amplitude=5pt}] (18.8,-1.25) -- (13.2,-1.25);
\node[align=center] at (16,-2.55) {Block$_3$\\ $\tfrac12 \gamma\Delta$ Hashwork};
\end{tikzpicture}

\caption{Hashwork is the product of active hashpower and active mining time.
Each block has length $\Delta$. The three blocks are alternative scenarios. The network is idle during the
white gap and active during the hatched interval. Block$_1$ has
$2\gamma$ hashpower active for $\Delta/2$, giving $1 \gamma\Delta$ Hashwork.
Block$_2$ has half the hashpower for the same active time, giving
$\tfrac12 \gamma\Delta$ Hashwork. Block$_3$ keeps the same hashpower as Block$_1$
but halves the active time, also giving $\tfrac12 \gamma\Delta$ Hashwork. The gap length determines the hashpower share needed to catch the honest chain. An attacker who mines throughout a block catches up once its hashwork over the whole block exceeds the honest hashwork, which requires just over $33\%$ of the combined hashpower in Blocks 1 and 2 and just over $20\%$ in Block 3. Hashwork determines the energy needed to match the honest network's block production rate.}
\label{fig:hashwork}
\end{figure}
These three metrics are not independent. The difficulty-adjustment mechanism couples
$\mu$ to $\gamma$ and to the gap length through \eqref{eq:mw}. Changes in the cost structure, fee levels, or
block subsidy impact all three metrics simultaneously, and not always in the same direction.
Understanding these interactions is the central goal of this paper.
\section{Aggregate Hashpower and the Vanishing Active Phase}
\label{sec:hashpower}
We now study how the equilibrium composition of the round
responds to growth in total hashpower $\gamma$, holding the invariant
winning probability $w^*$ determined by the difficulty-adjustment mechanism of
\cref{sec:daa}.
\subsection{When Does a Mining Gap Form?}
From \eqref{eq:s1common}, $s_1(w) > 0$ if and only if
$\bar{c}_{\mathrm{op}}/w > R_b + s_b r_0$, meaning that operating instantly with only the block subsidy and Type-0 transactions does not cover the operating cost. In the absence of a gap,
$m(w) = 1/(w\gamma)$ and the invariant condition $m(w^*) = B$ gives
$w^* = 1/(\gamma B)$. For a mining gap to occur, we require $s_1(w^*) > 0$. Substituting $w^* = 1/(\gamma B)$ into the gap condition $s_1(w^*) > 0$:
\begin{equation}
\gamma \;>\; \gamma_0 \;:=\; \frac{R_b + s_b\,r_0}{\bar{c}_{\mathrm{op}}\,B}
\label{eq:gamma0}
\end{equation}

For $\gamma \leq \gamma_0$, the baseline block revenue suffices to cover
operating costs and no gap forms. For $\gamma > \gamma_0$, a gap exists
at equilibrium. We assume $\gamma > \gamma_0$ for the remainder of this
section.
\subsection{The Aggregate Active Block-Finding Rate vs The Aggregate Active Hashpower}
Two equations jointly determine the equilibrium pair $(\mu^*, s_1)$ for
a given $\gamma$. The first is the expected block length condition
\eqref{eq:mw} evaluated at the invariant winning probability $w^* = \mu^*/\gamma$:
\begin{equation}
\frac{s_1}{\lambda} + \frac{1 - e^{-K(s_b - s_1)}}{K\lambda} = B,
\label{eq:sysA2}
\end{equation}
where $K$ is the unique positive root of
\eqref{eq:char} at $\mu = \mu^*$. The second is the threshold
\eqref{eq:s1common} rewritten with $w = \mu^*/\gamma$:
\begin{equation}
s_1 \;=\; \frac{\bar{c}_{\mathrm{op}}\,\gamma/\mu^* - R_b - s_b\,r_0}{r_1 - r_0}.
\label{eq:sysB2}
\end{equation}

Any change in $\gamma$ affects both equations simultaneously, and the
pair $(\mu^*, s_1)$ adjusts to satisfy both at once. Solving this system shows that
the aggregate active block-finding rate $\mu^*$ is strictly increasing in $\gamma$
(\cref{prop:hashpower}). Intuitively, more hashpower leads the difficulty-adjustment
mechanism to lower the winning probability. Each unit of hashpower then finds blocks
less often, so its reward rate $w\,R(X,t)$ covers the operating cost only at a higher
fee level, and miners wait for a higher threshold $s_1$ before they start. Since the
average round length stays at $B$, the longer wait leaves less time for the active
phase, and $\mu^*$ must rise. Through the
identity $\mathbb{E}[T_g] = B - 1/\mu^*$, every increment in $\gamma$
that raises $\mu^*$ also lengthens the expected gap.
\subsection{The Vanishing Active Phase}
Since $\mu^*$ is strictly increasing in $\gamma$, it either converges to
a finite limit or grows without bound. We show it must grow without
bound. From \eqref{eq:sysA2}, the second term on the left is strictly
positive, so $s_1/\lambda < B$, giving $s_1 < \lambda B$. Substituting
\eqref{eq:sysB2} and rearranging:
\begin{equation}
\frac{\gamma}{\mu^*} \;<\;
\frac{(r_1-r_0)\,\lambda B + R_b + s_b\,r_0}{\bar{c}_{\mathrm{op}}}
\;=:\; C \;<\; \infty.
\label{eq:Cbound2}
\end{equation}

The ratio $\gamma/\mu^*$ is bounded above by a finite constant that does
not depend on $\gamma$. Therefore $\mu^* > \gamma/C \to \infty$ as
$\gamma \to \infty$, and
\begin{equation}
\mathbb{E}[T_g] \;=\; B - \frac{1}{\mu^*} \;\longrightarrow\; B.
\label{eq:vanishing}
\end{equation}

The entire round is asymptotically consumed by the mining gap,
and the active mining phase vanishes.
\begin{proposition}
\label{prop:hashpower}
Suppose $\lambda B<s_b$. (i) The invariant winning probability has a mining gap,
$s_1(w^*)>0$, if and only if $\gamma>\gamma_0$. (ii) For $\gamma>\gamma_0$, the
aggregate block-finding rate $\mu^*$, the mining threshold $s_1(w^*)$, the expected gap
$\mathbb{E}[T_g]=B-1/\mu^*$ and the hashwork $\mathcal{W}=\gamma/\mu^*$ are continuous
and strictly increasing in $\gamma$. (iii) As $\gamma\to\infty$, $\mu^*\to\infty$,
$s_1(w^*)\to\lambda B$, $\mathbb{E}[T_g]\to B$, and $\mathcal{W}$ increases to $C$.
\end{proposition}
This result has a direct security
interpretation. During the gap, all
honest participants are idle. An adversary who mines continuously during
the gap accumulates blocks uncontested. As the gap approaches $B$, the
adversary's uncontested window approaches the full duration of every
round. The classical majority-hashpower security threshold then loses its
meaning. An adversary with hashpower $\gamma_{\mathrm{adv}}$ who mines
continuously finds blocks at rate $w^*\gamma_{\mathrm{adv}} =
\mu^*\gamma_{\mathrm{adv}}/\gamma > \gamma_{\mathrm{adv}}/C$ by \eqref{eq:Cbound2},
while the honest network finds one block per $B$ on average. The adversary
therefore outpaces the honest chain once $\gamma_{\mathrm{adv}} \ge C/B$. Since $C$
does not depend on $\gamma$, this requirement becomes a vanishing fraction of the
total hashpower as $\gamma$ grows. The exact requirement is
$\gamma_{\mathrm{adv}} > \mathcal{W}/B$, where $\mathcal{W}=\gamma/\mu^*$ is the hashwork
defined in \eqref{eq:hashwork_def}. By \eqref{eq:sysB2},
$\bar c_{\mathrm{op}}\mathcal{W}=R_b+s_b r_0+(r_1-r_0)s_1$, and by
\cref{prop:hashpower}, $s_1$ and with it $\mathcal{W}$ rise with $\gamma$ while
$\mathcal{W}$ stays below $C$. The absolute hashpower that an
attacker needs therefore rises with $\gamma$, even though its share of the total
hashpower vanishes.

The limit in \eqref{eq:vanishing} presumes that the difficulty-adjustment mechanism
settles at the invariant winning probability $w^*$. By \cref{sec:daa}, this is the case
for large $\gamma$ when $(r_1-r_0)\lambda B > R_b+s_b r_0$, that is, when the fee
premium makes up most of the expected block revenue. When the baseline block revenue is
larger, the invariant winning probability is unstable for large $\gamma$, and
\citet{kimSeo2026mind} show for a mempool without carryover that the winning
probability then oscillates across epochs. The next
section identifies the economic force that keeps $\gamma$ finite.
\section{Free Entry and the Equilibrium Hashpower}
\label{sec:entry}
We now let aggregate hashpower be determined by free entry. Each unit of
hashpower $i \in I$ bears a capital cost rate $c_{\mathrm{cap}}(i) > 0$, paid continuously
whether the unit is active or idle. This cost is sunk for the activation
decision in \cref{sec:strategies}, but it matters for entry.
We consider the case in which every participant pays the same price for mining hardware, such as the application-specific integrated circuits (ASICs) built for Bitcoin mining.
As in \cref{ass:commoncost}, the restriction needs to hold only for $\eta$-almost every participant, that is, for all participants except possibly a group with zero total hashpower.
\citet{tsabary2018} also treat all mining rigs as identical in cost, and the free-entry analysis of \citet{budish2025trust} uses a single cost per unit of computational power that includes the rental cost of capital.
Under this scenario, we impose the following restriction.
\begin{assumption}[Common Capital Cost]
\label{ass:capcost}
There exists $\bar{c}_{\mathrm{cap}}>0$ such that $c_{\mathrm{cap}}(i)=\bar{c}_{\mathrm{cap}}$ for $\eta$-almost every participant.
\end{assumption}

Write
\[
\alpha := r_1-r_0,
\qquad
R_0 := R_b+s_b r_0.
\]

Here $R_0$ is the baseline block revenue from the subsidy and the low-tier
fee component, while $\alpha$ is the premium earned by Type-1 transactions
over Type-0 transactions.

At the invariant winning probability, the mean round length is $B$. Flow balance
implies that the expected number of Type-1 transactions included per block
is $\lambda B$. Hence the expected total block revenue is
\begin{equation}
R_0+\alpha\lambda B.
\label{eq:totalrev}
\end{equation}

Since each unit of hashpower receives a proportional share of block rewards,
the expected revenue per unit of hashpower is
\begin{equation}
\mathcal{R}(\gamma)
\;:=\;
\frac{R_0+\alpha\lambda B}{\gamma}.
\label{eq:unitrevenue}
\end{equation}

The expected cost per unit of hashpower is
\begin{equation}
\mathcal{C}(\gamma)
\;:=\;
\bar{c}_{\mathrm{cap}}B+\frac{\bar c_{\mathrm{op}}}{\mu^*(\gamma)}.
\label{eq:unitcost}
\end{equation}

The first term is capital cost over one target round length. The second
term is operating cost over the active mining phase, whose expected length
is $1/\mu^*(\gamma)$.

Free entry means that hashpower keeps joining the network while mining earns a
positive expected profit and leaves while it earns a negative one, so in the long run
expected revenue equals expected cost. \citet{budish2025trust} and
\citet{ciaian2021security} use the same zero-profit condition to model the
permissionless entry of miners. \citet{prat2021equilibrium} show that free entry places
an upper bound on mining revenue and use it to predict the network's hashpower. The
free-entry condition is
\begin{equation}
\frac{R_0+\alpha\lambda B}{\gamma}
=
\bar{c}_{\mathrm{cap}}B+\frac{\bar c_{\mathrm{op}}}{\mu^*(\gamma)}.
\label{eq:freeentry}
\end{equation}

If the left side exceeds the right side, mining is still profitable in expectation and additional hashpower enters. If the
left side is below the right side, mining is unprofitable and hashpower exits.

When a gap exists, the mining threshold is positive, and \eqref{eq:sysB2} gives
$\bar c_{\mathrm{op}}\gamma/\mu^*(\gamma)=R_0+\alpha s_1(\gamma)$. Substituting this
into \eqref{eq:freeentry} removes the operating-cost term and gives the equivalent
condition
\begin{equation}
\frac{\alpha\bigl(\lambda B-s_1(\gamma)\bigr)}{\gamma}
=
\bar{c}_{\mathrm{cap}}B.
\label{eq:entrygap}
\end{equation}

This free-entry condition has a simple interpretation. The fee premium earned
above the activation threshold must cover the capital cost of hashpower.
The left side of \eqref{eq:entrygap} is strictly decreasing in $\gamma$.
In \cref{sec:hashpower}, it is shown that higher aggregate hashpower increases the mining threshold $s_1(\gamma)$, and the limit is $\lambda B$. Since the left side of \eqref{eq:entrygap} converges to zero as
$\gamma\to\infty$, free entry pins down a finite equilibrium hashpower.
Thus, capital cost prevents the unbounded growth identified in
\cref{sec:hashpower}.
\begin{proposition}
\label{prop:entry}
Suppose $\lambda B<s_b$. The expected profit per unit of hashpower, the left side of
\eqref{eq:freeentry} minus its right side, is continuous and strictly decreasing in
$\gamma$. It tends to infinity as $\gamma\to0$ and to $-\bar c_{\mathrm{cap}}B<0$ as
$\gamma\to\infty$. Hence there is a unique free-entry equilibrium hashpower
$\gamma^*\in(0,\infty)$, hashpower enters below $\gamma^*$ and exits above it, and in the
gap regime $\gamma^*$ satisfies \eqref{eq:entrygap}. If the capital cost were zero,
expected profit would be positive for every $\gamma$, and no finite equilibrium would
exist.
\end{proposition}
\Cref{fig:equilibrium} illustrates the long-run determination of the equilibrium hashpower.
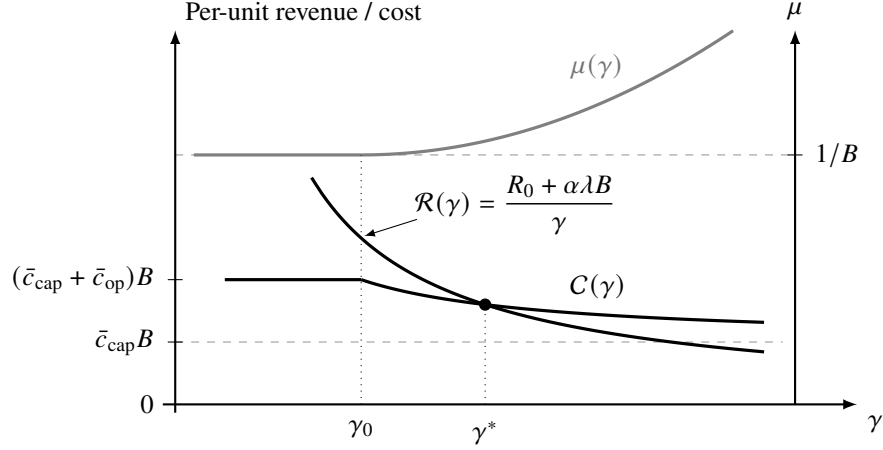
\begin{figure}[tb]
\centering
\begin{tikzpicture}[x=0.82cm, y=0.55cm, >=latex, every node/.style={font=\small}]
\draw[dashed,black!40] (0,1.5) -- (9.8,1.5);   %
\draw[dashed,black!40] (0,6)   -- (10,6);       %
\draw[very thick,black!50] (0.3,6) -- (3,6);
\draw[very thick,black!50,domain=3:9,samples=40,smooth]
plot (\x,{6 + 0.083*(\x-3)*(\x-3)});
\node[anchor=south,black!55] at (6.8,7.6) {$\mu(\gamma)$};
\draw[very thick] (0.8,3) -- (3,3);
\draw[very thick,domain=3:9.5,samples=40,smooth] plot (\x,{1.5 + 4.5/\x});
\node[anchor=south west] at (6.2,2.3) {$\mathcal{C}(\gamma)$};
\draw[very thick,domain=2.2:9.5,samples=45,smooth] plot (\x,{12/\x});
\node[anchor=west] at (3.7,4.8)
{$\mathcal{R}(\gamma)=\dfrac{R_0+\alpha\lambda B}{\gamma}$};
\draw[->] (3.9,4.45) -- (3.05,4.05);
\fill (5,2.4) circle (2.3pt);
\draw[dotted] (3,0) -- (3,6);      %
\draw[dotted] (5,0) -- (5,2.4);    %
\node[below] at (3,-0.15) {$\gamma_0$};
\node[below] at (5,-0.15) {$\gamma^{*}$};
\draw[->,thick] (0,-0.2) -- (0,9) ;
\node[anchor=south west] at (0,9) {Per-unit revenue / cost};
\draw (-0.12,1.5) -- (0.12,1.5);
\node[anchor=east] at (-0.18,1.5) {$\bar{c}_{\mathrm{cap}}B$};
\draw (-0.12,3)   -- (0.12,3);
\node[anchor=east] at (-0.18,3) {$(\bar{c}_{\mathrm{cap}}+\bar{c}_{\mathrm{op}})B$};
\node[anchor=east] at (-0.18,0) {$0$};
\draw[->,thick] (10,0) -- (10,9) node[above] {$\mu$};
\draw (9.88,6) -- (10.12,6);
\node[anchor=west] at (10.16,6) {$1/B$};
\draw[->,thick] (-0.2,0) -- (11,0) node[below right] {$\gamma$};
\end{tikzpicture}
\caption{Long-run determination of the equilibrium hashpower. The lower
curves use the left axis. The block revenue per unit, $\mathcal{R}(\gamma)$,
falls with aggregate hashpower $\gamma$ and approaches zero, while the cost per unit,
$\mathcal{C}(\gamma)=\bar{c}_{\mathrm{cap}}B+\bar{c}_{\mathrm{op}}/\mu$, is
constant before a gap appears and then falls toward
$\bar{c}_{\mathrm{cap}}B$. The two curves meet at the unique free-entry
equilibrium hashpower $\gamma^{*}$. The gray curve uses the right axis. The
aggregate active block-finding rate $\mu(\gamma)$ is $1/B$ when mining is
continuous and increases once the gap appears. Both changes occur at the
same threshold $\gamma_0$. Below $\gamma_0$, there is no gap and
$1/\mu=B$. Above $\gamma_0$, the gap grows, $\mu$ increases, and block revenue and cost
per unit both fall.}
\label{fig:equilibrium}
\end{figure}
\section{Preventing the Mining Gap}
\label{sec:prevent}
We next ask when the free-entry equilibrium has no mining gap. In the no-gap regime, mining starts immediately in every round. Hence
$\mu^*=1/B$, and the active phase lasts the whole target round length in
expectation. The free-entry condition \eqref{eq:freeentry} becomes
\begin{equation}
\frac{R_0+\alpha\lambda B}{\gamma}
=
(\bar{c}_{\mathrm{cap}}+\bar c_{\mathrm{op}})B.
\label{eq:nogapentry}
\end{equation}

A no-gap equilibrium exists if and only if the resulting equilibrium hashpower
$\gamma^*$ does not exceed the threshold $\gamma_0=R_0/(\bar c_{\mathrm{op}}B)$ from
\eqref{eq:gamma0}. Solving \eqref{eq:nogapentry} for $\gamma^*$ shows that this is
equivalent to
\begin{equation}
\frac{\bar c_{\mathrm{op}}}{\bar c_{\mathrm{cap}}}
\;\leq\;
\frac{R_0}{\alpha\,\lambda B}.
\label{eq:nogapcondition}
\end{equation}

Condition \eqref{eq:nogapcondition} is the no-gap condition. The left side
is the ratio of operating cost to capital cost. The right side is the ratio of baseline revenue to the expected fee premium flow over one target round length. A high baseline reward or a high capital cost keeps the network in the no-gap regime. A high fee
premium or a high operating cost pushes the system toward a gap. Since
$\alpha=r_1-r_0$, the parameter effects in \cref{tab:prevent,tab:reduce} hold $r_1$ fixed
when $r_0$ increases, so a higher $r_0$ also lowers $\alpha$, and they hold $r_0$ fixed
when $\alpha$ increases.
\begin{proposition}
\label{prop:nogap}
Suppose $\lambda B<s_b$. The free-entry equilibrium has no mining gap,
$\gamma^*\le\gamma_0$, if and only if \eqref{eq:nogapcondition} holds. In that case
$\gamma^*=(R_0+\alpha\lambda B)/\bigl((\bar c_{\mathrm{cap}}+\bar c_{\mathrm{op}})B\bigr)$,
and the effects of each parameter on \eqref{eq:nogapcondition} and on $\gamma^*$ are those
in \cref{tab:prevent}.
\end{proposition}

\begin{table}[t]
\centering
\caption{Parameter effects on gap prevention and equilibrium hashpower.}
\label{tab:prevent}
\begin{tabular}{@{}lcc@{}}
\toprule
Parameter increase
& \begin{tabular}[c]{@{}c@{}}Prevents\\ mining gap\end{tabular}
& \begin{tabular}[c]{@{}c@{}}Increases equilibrium\\ hashpower $\gamma^*$\end{tabular} \\
\midrule
Block subsidy $R_b \uparrow$ & \checkmark & \checkmark \\
Low-tier fee $r_0 \uparrow$ & \checkmark & \checkmark \\
Fee premium $\alpha \uparrow$ & $\times$ & \checkmark \\
Capital cost $\bar{c}_{\mathrm{cap}} \uparrow$ & \checkmark & $\times$ \\
Operating cost $\bar c_{\mathrm{op}} \uparrow$ & $\times$ & $\times$ \\
\bottomrule
\end{tabular}
\end{table}

\Cref{tab:prevent} highlights the role of capital cost. Higher capital cost reduces
entry, and this can eliminate the mining gap. This protection comes at the
cost of lower equilibrium aggregate hashpower. By contrast, a higher block subsidy
both helps prevent the gap and supports a larger equilibrium aggregate hashpower.
\section{Reducing the Mining Gap}
\label{sec:reduce}
Suppose now that the free-entry equilibrium lies in the gap regime. The
expected gap is
\begin{equation}
\mathbb{E}[T_g]
=
B-\frac{1}{\mu^*}.
\label{eq:gap_freeentry}
\end{equation}

Thus, reducing the gap is equivalent to reducing the equilibrium aggregate
block-finding rate $\mu^*$. We also track the two security metrics introduced in \cref{sec:metrics}: the equilibrium aggregate hashpower $\gamma^*$ and the hashwork $\mathcal{W} = \gamma^*/\mu^*$ from \eqref{eq:hashwork_def}.
The parameter effects follow from the two equilibrium equations in the
gap regime: the invariant block-time condition from \cref{sec:daa} and the
free-entry condition \eqref{eq:entrygap}. Total differentiation of this
two-equation system gives the signs in \cref{tab:reduce}. The calculation
uses the fact in \cref{sec:hashpower} that higher aggregate
hashpower raises $\mu^*$ and lengthens the expected gap.
\begin{proposition}
\label{prop:statics}
Suppose $\lambda B<s_b$ and that the free-entry equilibrium has a mining gap. The effects
of each parameter on the expected gap, the hashwork $\mathcal{W}$ and the equilibrium
hashpower $\gamma^*$ are those in \cref{tab:reduce}. The effect of $r_0$ on $\gamma^*$ is
positive if $\Gamma<\Gamma^\star$ and negative if $\Gamma>\Gamma^\star$, where
$\Gamma:=\bar c_{\mathrm{op}}/\bar c_{\mathrm{cap}}$ and $\Gamma^\star$ is given in
\cref{app:statics}.
\end{proposition}
\begin{table}[t]
\centering
\caption{Parameter effects in the gap regime. A check mark indicates an
improvement in the listed metric. $\pm$ indicates an ambiguous effect.}
\label{tab:reduce}
\begin{tabular}{@{}lccc@{}}
\toprule
Parameter increase
& \begin{tabular}[c]{@{}c@{}}Reduces\\ mining gap\end{tabular}
& \begin{tabular}[c]{@{}c@{}}Increases\\ hashwork $\mathcal{W}$\end{tabular}
& \begin{tabular}[c]{@{}c@{}}Increases equilibrium\\ hashpower $\gamma^*$\end{tabular} \\
\midrule
Block subsidy $R_b \uparrow$ & \checkmark & \checkmark & \checkmark \\
Low-tier fee $r_0 \uparrow$ & \checkmark & \checkmark & $\pm$ \\
Fee premium $\alpha \uparrow$ & $\times$ & \checkmark & \checkmark \\
Capital cost $\bar{c}_{\mathrm{cap}} \uparrow$ & \checkmark & $\times$ & $\times$ \\
Operating cost $\bar c_{\mathrm{op}} \uparrow$ & $\times$ & $\times$ & $\times$ \\
\bottomrule
\end{tabular}
\end{table}

The block subsidy has the most favorable effect. Increasing $R_b$ lowers
the threshold required for honest miners to activate. The gap shrinks,
honest hashwork per block rises, and the equilibrium aggregate hashpower rises.

Increasing the low-tier fee $r_0$ decreases the gap and increases hashwork, but has an ambiguous effect on the equilibrium aggregate hashpower. The direction depends on the cost structure of the network through a threshold $\Gamma^\star$ derived in \cref{app:statics}.

A higher fee premium $\alpha$ lengthens the gap but raises both the hashwork and the
equilibrium aggregate hashpower. More fee revenue draws more hashpower into the network,
and, as in \cref{sec:hashpower}, the difficulty-adjustment mechanism then shortens the
active phase.

Capital cost has some tradeoffs. Higher $\bar{c}_{\mathrm{cap}}$ reduces
hashpower entry. This lowers the equilibrium aggregate block-finding rate and shortens the gap, but
it also reduces both equilibrium aggregate hashpower and honest hashwork per block.

A higher operating cost, for example from higher electricity prices, works against all
three metrics. It lengthens the gap and lowers both the hashwork and the equilibrium
aggregate hashpower.
\section{A Numerical Illustration for Bitcoin}
\label{sec:numerical}
This section uses Bitcoin data to show the size and the timing of the effects derived in
the previous sections. We compare two settings. In the first, the aggregate hashpower
stays at its 2026 level, so that no miner enters or leaves the network after 2026. This
is the setting of \cref{sec:daa,sec:hashpower}, in which the aggregate hashpower is given.
In the second, hashpower enters and exits freely, so that the network is at the
free-entry equilibrium of \cref{sec:entry} in every year. For each setting we follow the
block subsidy along Bitcoin's halving schedule from 2026 to 2070 under a few scenarios for
the bitcoin price and transaction fees. These scenarios illustrate the model. They are not
forecasts of prices, fees or mining technology.

\subsection{Inputs Without Units}
\label{sec:numerical-inputs}
The results below depend on the parameters only through a few ratios without units.
Measuring transactions in multiples of the block capacity $s_b$, time in multiples of the
target round length $B$, and amounts of money in multiples of the fee premium per target
round $(r_1-r_0)\lambda B$ leaves the equations of \cref{sec:backlog,sec:daa,sec:hashpower,sec:entry}
unchanged and reduces their inputs to
\begin{equation}
\nu \;:=\; \frac{R_b+s_b r_0}{(r_1-r_0)\lambda B},
\qquad
\xi \;:=\; \frac{\lambda B}{s_b},
\label{eq:unitfree}
\end{equation}
together with one more ratio that depends on the setting. The ratio $\nu$ compares the
baseline block revenue with the fee premium that Type-1 transactions add over a target
round length. It is the right side of the no-gap condition \eqref{eq:nogapcondition}, and
$-1-\nu$ is the limit of the elasticity in \cref{prop:stability}. The ratio $\xi$ is the
share of the block capacity that Type-1 arrivals fill over a target round length, and
$\xi<1$ because $\lambda B<s_b$.

With fixed hashpower, the third input is $\gamma/\gamma_0$, the aggregate hashpower divided
by $\gamma_0$ of \eqref{eq:gamma0}, which is the largest aggregate hashpower at which no gap
forms. A mining gap forms if and only if $\gamma/\gamma_0>1$ (\cref{prop:hashpower}). With
free entry, the third input is the cost ratio $\Gamma=\bar c_{\mathrm{op}}/\bar c_{\mathrm{cap}}$
of \cref{prop:statics}, and a mining gap forms if and only if $\Gamma>\nu$
(\cref{prop:nogap}).

In the free-entry setting, the scenarios hold only the ratio $\Gamma$ fixed. Both costs may change over time, for example through cheaper
electricity or more efficient machines, as long as they change in the same proportion.
Such a change alters how much hashpower the network attracts, but it leaves
$\gamma^*/\gamma_0$, the mining gap and its stability unchanged, because the revenue per
unit of hashpower adjusts together with the costs. In the fixed-hashpower setting, the
hashpower and its operating cost in dollars stay at their 2026 levels.

We report two security measures. The first is the expected mining gap
$\mathbb{E}[T_g]=B-1/\mu^*$ with $B=10$ minutes. The second is the share of the total
hashpower above which an adversary who mines without interruption finds blocks faster
than the honest network on average, so that its chain eventually overtakes the honest
chain with certainty. A smaller adversary can still succeed with positive probability
\citep{nakamoto2008}. By the requirement $\gamma_{\mathrm{adv}}>\mathcal W/B$ of
\cref{sec:hashpower} and $\mathcal W=\gamma/\mu^*$, this share equals $1/(1+\mu^*B)$. It is
one half without a gap and falls as the active phase $1/\mu^*$ shrinks.

\subsection{Estimating the Inputs}
\label{sec:numerical-estimation}
We estimate the inputs from Bitcoin block data of the six difficulty-adjustment periods
between block heights 957{,}600 and 969{,}695, which ran from 11 July to 3 October 2026
\citep{mempoolspace2026}. We write BTC for the unit of bitcoin and USD for United States
dollars. The block subsidy was $R_b=3.125$~BTC throughout, transaction fees averaged
0.0211~BTC per block, and blocks were 98\% full on average. Bitcoin measures the size of a
transaction in virtual bytes (vB), and a block holds at most $s_b=10^6$~vB
\citep{bitcoinDevBlockChain}. The two-tier mempool of the model is an approximation of the
many feerates of a real mempool, so the values below are estimates read through the model
and are prone to error. We report the range of each estimate and use one value in the
figures. \Cref{ec:fees} describes the estimation in detail.

From the block fees we estimate $s_b r_0$ and the fee premium $(r_1-r_0)\lambda B$. By
\eqref{eq:reward}, a block found $\tau$ minutes after its predecessor, in a round that
starts without Type-1 transactions waiting, earns the fees $s_b r_0+(r_1-r_0)\lambda\tau$.
We therefore use the 737 blocks whose predecessor was found within 2 minutes, which most
likely included every Type-1 transaction then waiting, and whose own $\tau$ is at most
5 minutes, before a block can fill with Type-1 transactions. A robust line through their
fees, the median of the slopes between pairs of blocks, gives the intercept
$s_b r_0=0.0045$~BTC and the premium $(r_1-r_0)\lambda B=0.0154$~BTC. Hence $\nu=203$ in
2026, and $\nu$ would be $0.29$ without the subsidy. When the predecessor must be found
within 1 or 3 minutes instead of 2, the value without the subsidy ranges from 0.23 to 0.32.

The block fees reveal only the product $(r_1-r_0)\lambda$ and not its two factors
separately, so we estimate $\lambda$ from the size of the transactions. We count a transaction as Type-1 if its feerate exceeds $r_0$ by
more than 20\%, and $\lambda$ is the growth of the Type-1 size of a block with $\tau$ in
the same 737 blocks. This gives $\lambda=73{,}000$~vB per minute and $\xi=0.73$, which
satisfies $\lambda B<s_b$. The estimate stays between 0.70 and 0.75 for every margin above
$r_0$ from 5\% to 100\%. A margin above about 120\% would put the barrier above 1 satoshi per vB
(one satoshi is $10^{-8}$~BTC), the feerate that many transactions pay, and these
transactions accumulate in the same way as Type-1 transactions. Two checks support the
two-tier reading. The fees paid below the barrier fall with $\tau$ at the rate
$r_0\lambda$, as they should when each arriving Type-1 transaction displaces a Type-0 one.
A regression of the block fees on the Type-1 size, which removes the effect of any Type-1
transactions left over from earlier rounds, gives $\nu=0.29$ to 0.30 without the subsidy.

The cost ratio $\Gamma$ cannot be read from the blockchain, so we combine three sources.
First, \citet{budish2025trust} uses a capital share of 0.4 in the cost of mining, which
gives $\Gamma=0.6/0.4=1.5$. Second, without a gap the free-entry condition
\eqref{eq:nogapentry} states that the revenue per unit of hashpower equals its operating
cost plus its capital cost. Over the data window, the block revenue was 35.8~USD per day
for each PH/s of hashpower, where one PH/s is $10^{15}$ hash attempts per second and the
network hashpower is implied by the difficulty. We approximate the operating cost by the
cost of the electricity that this hashpower draws, between 17.4 and 18.3 kilowatts
including cooling and power conversion. The lower value adds the 10\% overhead of \citet{cbeci2026} to the 15.8
joules per $10^{12}$ hash attempts of a large listed miner's machines
\citep{coinshares2026q2}. The upper value divides the network's power demand estimated by
\citet{cbeci2026} by its hashpower. At 0.05~USD per kilowatt-hour, the value used by
\citet{capponi2023mining}, the operating cost is 21 to 22~USD per day per PH/s. The capital
cost is the remaining 14 to 15~USD, so $\Gamma$ lies between 1.4 and 1.6. Third, an
Antminer S21 XP uses 13.5 joules per $10^{12}$ hash attempts and was priced at 25.70~USD
per TH/s, where one TH/s is $10^{12}$ hash attempts per second \citep{hashrateindex2026}.
With the same 10\% overhead and its capital rented at 50\% per year, as in
\citet{budish2025trust}, this machine has $\Gamma=0.5$. We use the cost ratio $\Gamma=1.5$
as the estimate, and $\Gamma=0.5$, $\Gamma=3$ and $\Gamma=10$ as a low, a high and a very
high cost ratio. Because $\Gamma=\bar c_{\mathrm{op}}/\bar c_{\mathrm{cap}}$, a higher cost
ratio can come from a higher operating cost or from a lower capital cost. Our focus is the
second case, in which mining hardware becomes cheap relative to the electricity it uses.
By cheap hardware we mean a capital cost that is small relative to the operating cost,
which makes $\Gamma$ large.

For the fixed-hashpower setting we need $\gamma/\gamma_0$ in 2026, the aggregate hashpower
divided by the largest hashpower at which no gap forms. By \eqref{eq:gamma0}, this ratio
equals the operating cost of the network over a target round divided by the baseline
block revenue $R_b+s_b r_0$ in USD. As above, we approximate the operating cost by the
cost of the electricity that the network uses. The difficulty implies an aggregate
hashpower of 915~EH/s, that is, $9.15\times10^{20}$ hash attempts per second. At the power
use above, it costs 133 to 140 thousand USD of electricity per target round, while
$R_b+s_b r_0$ was worth 227 thousand USD at the average price of 72{,}424~USD per BTC. The
ratio lies between 0.58 and 0.62, and we use 0.60. The free-entry equilibrium with
$\Gamma=1.5$ gives the same value. Both lie below one, so the model predicts no mining gap
in 2026, in line with \citet{kimSeo2026mind}, who find no deliberate idle periods in
Bitcoin's hashpower data for 2022 and 2023. \Cref{tab:inputs} summarizes the estimates.

\begin{table}[t]
\centering
\caption{Estimated inputs for 2026.}
\label{tab:inputs}
\begin{tabular}{@{}lrr@{}}
\toprule
Quantity & Value used & Range \\
\midrule
Baseline fee $s_b r_0$ per block & 0.0045 BTC & \\
Fee premium $(r_1-r_0)\lambda B$ per target round & 0.0154 BTC & \\
Ratio $\nu$ in 2026 & 203 & \\
Ratio $\nu$ without the subsidy & 0.29 & 0.23 to 0.32 \\
Type-1 share of block capacity $\xi$ & 0.73 & 0.70 to 0.75 \\
Cost ratio $\Gamma$ & 1.5 & 0.5 to 1.6 \\
Hashpower over the largest gap-free hashpower, $\gamma/\gamma_0$ & 0.60 & 0.58 to 0.62 \\
\bottomrule
\end{tabular}
\end{table}

\paragraph{Scenarios.} The block subsidy halves every 210{,}000 blocks
\citep{bitcoinDevBlockChain}, which we date to 2028, 2032 and every four years after. For
the price and the fees, we let one quantity change at a time from 2026. Over the past ten
years, the fees in a block, measured in USD, have tended to rise when the bitcoin price
rose and to fall when it fell. Over the whole decade, however, the fees in a block,
measured in BTC, fell from about 0.28~BTC to about 0.02~BTC, while the price of one
bitcoin rose from about 400~USD to about 72{,}000~USD \citep{mempoolspace2026}. The data
therefore do not tell whether fees will follow the price, so our scenarios include both
possibilities.
\begin{itemize}
\item S0, the baseline: the price, the fees and the costs stay at their 2026 values, and
only the subsidy halves.
\item S1: the price grows by 10\% per year, and fees stay fixed in USD, so fees in BTC fall.
\item S2: the price grows by 10\% per year, and fees follow the price, so fees in BTC stay
fixed.
\item S3: the price stays fixed, and fees grow by 10\% per year in USD and in BTC.
\end{itemize}
Under S2, the price and the fees in USD grow at the same rate, so the fees in BTC, and with
them the ratio $\nu$, follow the same path as under S0. The free-entry equilibrium depends on
the scenario only through $\nu$, $\xi$ and $\Gamma$, so its results under S2 are the same as
under S0. The share $\xi$ stays at 0.73 in all scenarios.

\subsection{Free Entry Postpones the Gap}
\label{sec:numerical-entry}
\Cref{fig:num-hashpower} shows $\gamma/\gamma_0$, the aggregate hashpower divided by the
largest hashpower at which no gap forms, in both settings. With fixed hashpower, $\gamma_0$
falls at every halving unless the bitcoin price rises enough to offset it. In the baseline
S0, the 2028 halving raises $\gamma/\gamma_0$ from 0.60 to 1.20, so a gap forms in 2028,
and $\gamma/\gamma_0$ reaches about 300 by 2070. A price that grows by 10\% per year (S1 and
S2) lowers the ratio between halvings, but each halving raises it again, and the gap forms
after the 2032 halving. Growing fees (S3) lower the ratio only after about 2050, once fees
make up a large part of the block revenue. With free entry and $\Gamma=1.5$, the ratio
$\gamma^*/\gamma_0$ rises only from 0.60 in 2026 to 2.5 by 2070 under S0 and S2, and it
crosses one in 2056. It crosses one in 2044 when fees grow (S3), and it stays below one when
fees fall in BTC (S1). The free-entry hashpower $\gamma^*$ and the threshold $\gamma_0$ both depend on the block
revenue relative to the costs of mining, so each of them can rise or fall over time as the
costs change. A halving moves them together, and their ratio $\gamma^*/\gamma_0$ changes
only through $\nu$, which is why it rises slowly.

\begin{figure}[t]
\centering
\includegraphics[width=0.92\textwidth]{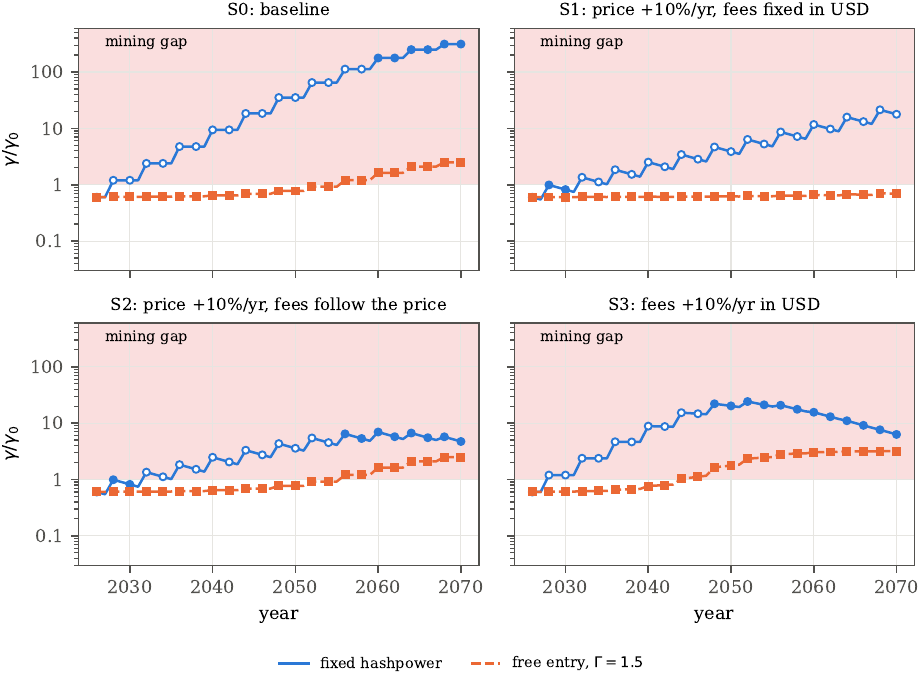}
\caption{Aggregate hashpower divided by the largest aggregate hashpower at which no gap
forms ($\gamma/\gamma_0$), from 2026 to 2070, for the four scenarios. A mining gap forms in
the shaded region where this ratio exceeds one. Solid lines keep the aggregate hashpower
and its operating cost at their 2026 levels (fixed hashpower). Dashed lines show the
free-entry equilibrium with $\Gamma=1.5$. Open markers show years in which the invariant
winning probability is unstable.}
\label{fig:num-hashpower}
\end{figure}

\subsection{Capital Costs Protect the Network}
\label{sec:numerical-capital}
\Cref{fig:num-capital} repeats the free-entry equilibrium for four cost ratios. A higher
$\Gamma$ means a lower capital cost relative to the operating cost. With fees constant in
BTC (S0 and S2), a gap forms in 2064 with $\Gamma=0.5$, in 2056 with $\Gamma=1.5$, in 2052
with $\Gamma=3$ and in 2044 with $\Gamma=10$. Growing fees (S3) bring each date forward, to
2052, 2044, 2043 and 2037. Fees that fall in BTC (S1) keep the network free of a gap through
2070 for every $\Gamma$ except 10, for which the gap forms in 2064. Without the subsidy,
$\nu$ equals $s_b r_0/((r_1-r_0)\lambda B)=0.29$ in every scenario, because the scenarios
change both fee tiers in the same proportion. As the subsidy shrinks toward zero, a gap
therefore forms in every scenario whenever $\Gamma$ exceeds 0.29, and none of the scenarios
prevents it. A low cost ratio delays the gap and keeps it short, as
\cref{sec:entry,sec:prevent} predict.

\begin{figure}[t]
\centering
\includegraphics[width=\textwidth]{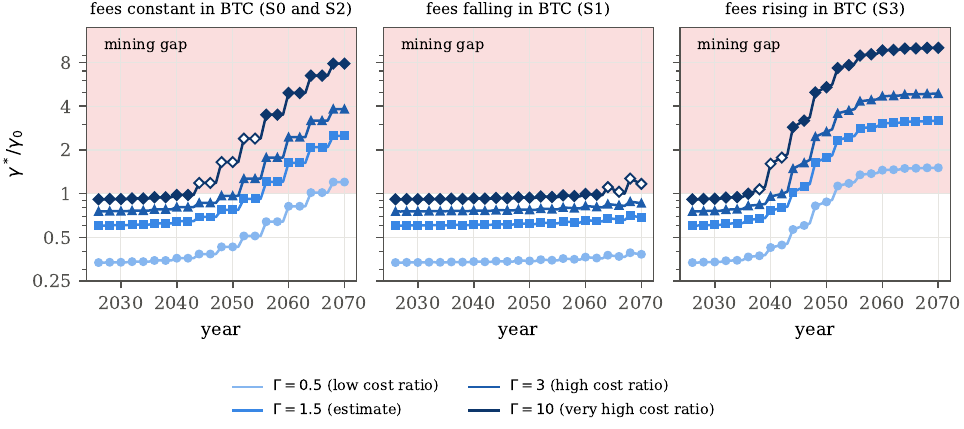}
\caption{Free-entry hashpower divided by the largest aggregate hashpower at which no gap
forms ($\gamma^*/\gamma_0$), for four cost ratios
$\Gamma=\bar c_{\mathrm{op}}/\bar c_{\mathrm{cap}}$. A mining gap forms in the shaded region
where this ratio exceeds one. The panels differ in the path of fees in BTC. Open markers
show years in which the invariant winning probability is unstable.}
\label{fig:num-capital}
\end{figure}

\subsection{The Length of the Mining Gap}
\label{sec:numerical-gap}
\Cref{fig:num-gap} shows the expected mining gap and the share of the total hashpower that
an adversary needs, as defined in \cref{sec:numerical-inputs}. With fixed hashpower in the
baseline S0, the gap lasts 1.6 minutes after the 2028 halving, 5.8 minutes after the 2032
halving and 8.9 minutes by 2040. By 2070 it takes 9.9 minutes of the 10-minute round, and
the adversary needs only 1.1\% of the total hashpower. When the price and the fees grow together (S2), the gap forms in 2032
and peaks at 7.4 minutes in 2052, when the adversary needs 21\% of the total hashpower. With
free entry, the gap is much shorter. With $\Gamma=1.5$, it forms in 2056 and reaches 2.4
minutes by 2070, when the adversary needs 43\% of the total hashpower. With $\Gamma=3$, it
forms in 2052 and reaches 4.0 minutes by 2070, when the adversary needs 37\%. In the years
drawn with open markers, the invariant winning probability is unstable. In those years the
reported values describe the invariant point, while the winning probability would
oscillate around it, so the realized gap can differ.

\begin{figure}[t]
\centering
\includegraphics[width=\textwidth]{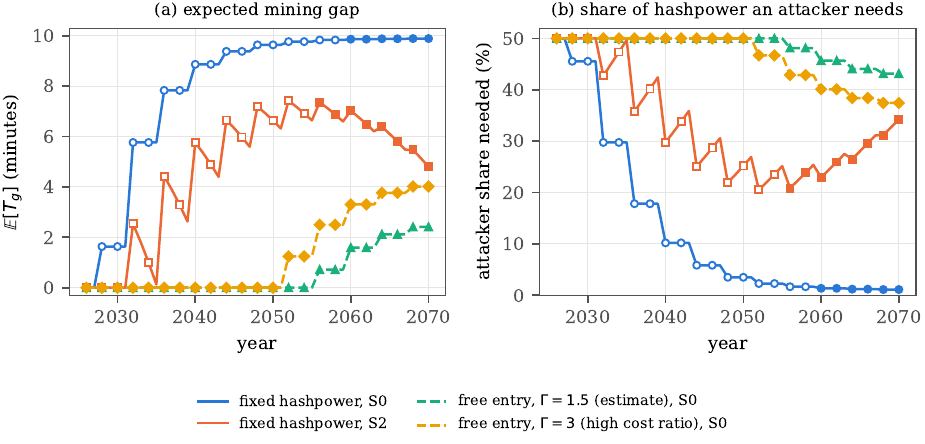}
\caption{(a) Expected mining gap $\mathbb{E}[T_g]$ in minutes. (b) Share of the total
hashpower above which an adversary who mines without interruption outpaces the honest
chain on average. Solid lines keep the aggregate hashpower at its 2026 level under S0 and
S2. Dashed lines show the free-entry equilibrium under S0 for the estimated cost ratio
$\Gamma=1.5$ and for a high cost ratio, $\Gamma=3$. Open markers show years in which the
invariant winning probability is unstable, and there the values describe the invariant
point.}
\label{fig:num-gap}
\end{figure}

\subsection{Stability Is Not Guaranteed}
\label{sec:numerical-stability}
\Cref{fig:num-stability} follows the stability of the invariant winning probability in the
free-entry equilibrium. The elasticity $\varepsilon(w^*)$ of \cref{prop:stability} depends
only on where the invariant point lies, which in the inputs of \cref{sec:numerical-inputs}
is given by $\nu$, $\gamma^*/\gamma_0$ and $\xi$. Panel (a) divides the plane of $\nu$ and
$\gamma^*/\gamma_0$ at $\xi=0.73$ into the region without a gap, the region with a gap and a
stable invariant point, and the region with a gap and an unstable one. The boundary between
the last two lies at $\nu$ between about 1.0 and 2.1. For a given $\Gamma$, every scenario moves
the free-entry equilibrium along the same curve, because the equilibrium depends on the
fees only through $\nu$. The scenario sets only the speed. As $\nu$ falls with the halvings,
each curve crosses into the gap region at $\nu=\Gamma$. With $\Gamma=1.5$, this crossing lies
in the stable region, and the invariant point stays stable in every year with a gap. With
$\Gamma=10$, the crossing lies in the unstable region, and the invariant point is unstable
from 2044 to 2055 under S0 and from 2037 to 2043 under S3, as panel (b) shows.

The carryover of Type-1 transactions matters here. With a mempool that resets after
every block, the elasticity would equal $-1-\nu$ at every aggregate hashpower (\cref{lem:reset}),
and the invariant point
with $\Gamma=1.5$ would be unstable for the first four years with a gap. With fixed
hashpower, the instability is far stronger. The gap then forms right after a halving, while
the subsidy still dominates the block revenue, and $\varepsilon(w^*)=-61$ when the gap forms
in 2028 under S0. The invariant point stays unstable until 2059. Our analysis is local, so
in unstable years we report the invariant point and do not model the dynamics around it.

\begin{figure}[t]
\centering
\includegraphics[width=\textwidth]{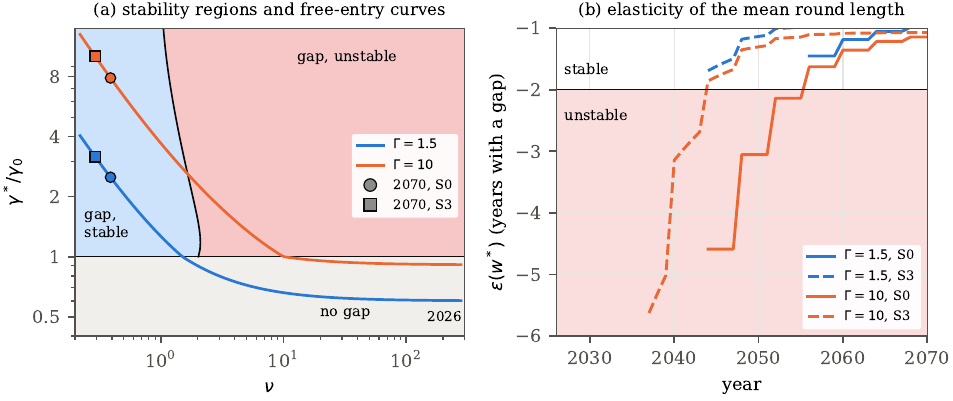}
\caption{Stability of the invariant winning probability in the free-entry equilibrium.
(a) Regions without a gap, with a gap and a stable invariant point, and with a gap and an
unstable invariant point, at $\xi=0.73$. The vertical axis is the free-entry hashpower
divided by the largest aggregate hashpower at which no gap forms. The curves show the
free-entry equilibrium as $\nu$ varies, for $\Gamma=1.5$ and $\Gamma=10$. The network
starts at the right end of each curve in 2026 and moves left as the subsidy falls. The circle
and the square on each curve show its position in 2070 under S0 and under S3. (b) Elasticity
$\varepsilon(w^*)$ of the mean round length in the years with a gap, for $\Gamma=1.5$ and
$\Gamma=10$, with solid lines for S0 and dashed lines for S3. The invariant point is
unstable below $-2$.}
\label{fig:num-stability}
\end{figure}

\subsection{Limits of the Illustration}
\label{sec:numerical-limits}
Three features of the illustration could change its results. First, we fit a two-tier fee
structure to a mempool in which transactions pay many different fees, so every estimate
carries the error of this simplification. Second, the fee regime can shift. In the
congested early months of 2024, the value of $\nu$ without the subsidy was about 1.0 instead
of 0.29, and a gap would then form only for $\Gamma$ above 1. Third, free entry describes
the long run. The capital cost of machines already bought is sunk, so the network adjusts
to a halving only gradually and can spend some time close to the fixed-hashpower setting.
\section{Discussion}
\label{sec:discussion}
Several modeling choices delimit the scope of the results.
First, the aggregate analysis uses a common operating cost and a common
capital cost up to null sets. Cost heterogeneity would spread activation
times across participants, and low-cost participants could mine earlier
than high-cost participants. The common-cost case makes the aggregate
activation rule a single threshold, which allows a clean stochastic
backlog analysis. Extending the analysis to heterogeneous costs is left for
future research. Second, the within-epoch backlog distribution is treated as stationary.
Each epoch has many rounds, and the recursion is ergodic under the
stability condition discussed earlier. The first rounds of an epoch still
inherit a transient state from the previous epoch. The main text abstracts
from this transient component, and measuring its effect on the mining gap is
left for future research. Third, the difficulty adjustment analysis is local. The map identifies the
invariant winning probability and the condition under which nearby perturbations
return to it. When the invariant point is unstable, the protocol can
oscillate across epochs. In the illustration of \cref{sec:numerical-stability}, this regime
appears when the cost ratio $\Gamma$ is very high, for up to 12 years after a gap first forms, and for
many of the years with a gap when hashpower cannot leave the network. That regime requires a separate analysis of the
joint evolution of backlog, entry, and difficulty, which we leave for future
research. Fourth, the fee process is intentionally simple. The two-tier fee market
separates baseline revenue from the premium attached to urgent
transactions. Real mempools
have many fee levels, strategic fee bidding, and endogenous demand.
\citet{seoKim2026priority}, for example, find that users bid lower fees when mining gaps
are longer, so the gap can feed back into the fee levels that we take as given. Those
features can be added, but they would obscure the main subject studied
in the paper. Incorporating them is a direction for future research. Fifth, the activation
decision compares the reward rate with the operating cost alone, so the only alternative
to mining is to stay idle. Hardware built for one PoW puzzle can also mine other coins
that use the same puzzle, such as Bitcoin Cash for Bitcoin's puzzle, and miners move
their hashpower between such coins as their relative profitability changes
\citep{kwon2019bitcoin}. With such an alternative, the cost of mining a coin also
includes the opportunity cost of the reward forgone on the other coin. Hashpower that
leaves Bitcoin during the gap would then secure another coin, so Bitcoin would still lose
it. Modeling this choice requires the joint dynamics of several coins and their
difficulty-adjustment mechanisms, which we leave for future research. Finally, the security interpretation is qualitative. During the gap, an adversary who mines continuously faces no honest
competition. This paper identifies the expected gap length as a primitive input for double-spending and selfish mining analysis.
Computing attack success probabilities from this input is left for future research.

Taken together, the results suggest that the three security metrics considered here, namely
the mining gap, the hashwork, and the equilibrium aggregate hashpower, do not move in concert under the
parameter changes we study (\cref{tab:prevent,tab:reduce}). Evaluating any protocol change or policy proposal therefore requires
specifying which metric is being optimized and at what cost to the others.
\section{Conclusion}
\label{sec:conclusion}
Bitcoin is moving toward the conditions under which a mining gap forms, and it moves
closer on a fixed schedule. The halving schedule cuts the block subsidy roughly every
four years until it reaches zero, and transaction fees must then pay for the security
of the network. Our results show that a smaller subsidy lowers the revenue that
miners can earn from the very start of a round, and this makes it harder to
keep the network free of a mining gap. Each halving therefore moves the network closer to a regime in which honest
miners sit idle at the start of every round while an attacker who keeps
mining faces no honest competition.

The central result is that aggregate hashpower can worsen the mining gap
problem. Once a gap forms, more hashpower means that blocks are found faster once
mining starts. The difficulty-adjustment mechanism keeps the average round length fixed, so
the time spent mining shrinks and the gap expands. In the limit of unbounded
hashpower, the gap approaches the full length of a round.

Capital cost prevents this limiting failure. Miners keep joining the network only
while expected revenue exceeds their capital and operating costs, and this pins
aggregate hashpower at a finite level. In our illustration with Bitcoin data, a capital
cost that is high relative to the operating cost delays the mining gap until 2064 and keeps
it below half a minute through 2070, while cheap hardware brings the gap forward to 2044 and lengthens it to six and a
half minutes by 2070, when the equilibrium hashpower is about eight times the largest
hashpower at which no gap forms (\cref{sec:numerical-capital}).
This result is a warning against efforts to make mining hardware cheap. Several
cryptocurrencies regard mining on ordinary, widely available hardware as a benefit.
Monero, for example, uses RandomX, a PoW puzzle designed to run efficiently on general-purpose
processors, so that specialized hardware offers little advantage and mining becomes
more decentralized \citep{tevadorRandomX}. Memory-hard puzzles have been proposed
with the same egalitarian aim \citep{biryukov2016egalitarian,biryukov2017equihash}.
In our model, lowering the capital cost of mining weakens the force that limits the
growth of hashpower, and with it the growth of the mining gap. Monero also uses a
different difficulty-adjustment algorithm and an adaptive block size
\citep[Ch.~7]{koe2020zeroToMonero}, whereas our
model assumes that every block has the same fixed capacity. Whether our results carry
over to Monero directly therefore requires a separate analysis. The balance between operating and capital costs that drives our results
is, however, present in any PoW network whose difficulty adjustment targets a fixed
average block time. Designers of such networks should weigh this balance before
pursuing cheaper mining hardware.

The block subsidy calls for the same care. Among the parameters we study, it is the
only one that shortens the gap while raising both the equilibrium hashpower and the
computational work behind each block. Cutting the subsidy therefore weakens all three
of the security measures we analyze at once. Our results should be read as an
optimistic limit. The model excludes forks and assumes that the difficulty-adjustment
mechanism settles at a steady puzzle difficulty, and both choices favor the network.
We therefore expect real networks to fare no better than our results suggest when the
subsidy shrinks. This caution applies to Bitcoin's halvings and to any other PoW
cryptocurrency that plans to reduce or end its subsidy.

A higher fee on every transaction offers a partial remedy. Raising the fee paid by the
cheapest transactions shortens the gap and increases the computational work behind
each block. When no gap exists, it helps prevent one in the same way as the block
subsidy. A network could adopt policies that raise this fee floor. The remedy has a
cost, since users then pay more for every transaction, and its effect on the
equilibrium hashpower depends on the cost structure of the network. It is a trade-off
between the cost of using the network and its security.

The analysis suggests that the composition of miner revenue matters as
much as its level. Revenue that is available immediately keeps miners
active. Revenue that accumulates only after delay can make waiting optimal, even though the resulting idle phase weakens the
network.

\ACKNOWLEDGMENT{The authors thank Damra Taylan Benli and Muhammed Serhan Hosta for their
insights and for many productive conversations that helped develop the ideas in this
paper.}
\bibliographystyle{plainnat}
\bibliography{references}
\ECSwitch
\ECHead{Proofs and Estimation Details}

\section{Notation}
\label{ec:notation}
In addition to the notation of the main text, the proofs use the following shorthand.
\begin{itemize}
\item $\alpha:=r_1-r_0$ is the fee premium and $R_0:=R_b+s_b r_0$ is the baseline block
revenue, as in \cref{sec:entry}. We also write $V:=R_0/\alpha$.
\item $D:=s_b-\lambda B$ is the block capacity that remains after the expected Type-1
arrivals of one target round length. It is positive whenever $\lambda B<s_b$.
\item $Q:=\lambda B-s_1$ and $\phi:=s_b-s_1=D+Q$.
\item $\theta:=\mu/\lambda$, $K$ is the positive root of \eqref{eq:char}, and
$\beta:=1-\theta s_b e^{-Ks_b}$.
\item $g(K,\phi):=(1-e^{-K\phi})/K$, with partial derivatives
\[
\begin{aligned}
g_K&=-\frac{1-e^{-K\phi}-K\phi\, e^{-K\phi}}{K^2}<0,\\
g_\phi&=e^{-K\phi}\in(0,1).
\end{aligned}
\]
The sign of $g_K$ follows from $e^{x}>1+x$ for $x>0$.
\item $\Gamma:=\bar c_{\mathrm{op}}/\bar c_{\mathrm{cap}}$ is the ratio of operating cost
to capital cost.
\end{itemize}
An asterisk marks a quantity evaluated at the invariant winning probability $w^*$, for
example $\mu^*=w^*\gamma$, $s_1^*=s_1(w^*)$ and $Q^*=\lambda B-s_1^*$.

\section{The Myopic Activation Rule}
\label{ec:objective}
\proof{Proof of \cref{prop:myopic}.}
Fix a profile $\mathbf{a}_{-i}$ under which rounds have a finite expected length, and let
$a_i$ be any strategy of participant $i$. Because $\eta(\{i\})=0$, changing $a_i$ changes
the integrand of \eqref{eq:active_hashpower} only on a null set, so $\gamma_{\mathrm{act}}$
is the same function of the history for every choice of $a_i$. The survival function
$\bar F$, the round lengths $T^{(j,n)}$, the backlogs $X^{(j,n)}$ and the winning
probabilities $w^{(j)}$ are determined by $\gamma_{\mathrm{act}}$, the transaction
arrivals and the mining outcomes. Their joint distribution, and with it the denominator
of $\rho_i$, therefore does not depend on $a_i$.

Consider a round with backlog $X$ and winning probability $w^{(j)}$, and let $a_i^*$
denote the rule \eqref{eq:myopic_rule}. Since $a_i(t)\in[0,1]$, every history and every
$t$ satisfy
\[
a_i(t)\,\pi_i(t)\;\le\;\max\{\pi_i(t),0\}\;=\;a_i^*(t)\,\pi_i(t).
\]
Multiplying by $\bar F(t\mid X,\mathbf{a}_{-i})>0$ and integrating over $t$ gives
$\Pi_i^{(j,n)}(a_i;\mathbf{a}_{-i})\le\Pi_i^{(j,n)}(a_i^*;\mathbf{a}_{-i})$. The right
side is finite because $\pi_i(t)\le w^{(j)}R_{\max}$ and $\int_0^\infty\bar F\,dt$ is the
expected round length. Taking expectations and summing over rounds and epochs, the
numerator of $\rho_i$ under $a_i$ is at most the numerator under $a_i^*$ for every $J$,
while the denominators coincide. Hence
$\rho_i(a_i;\mathbf{a}_{-i})\le\rho_i(a_i^*;\mathbf{a}_{-i})$. The profile
$\mathbf{a}_{-i}$ was arbitrary, so $a_i^*$ is weakly dominant. The dominance is weak
because any $a_i$ that differs from $a_i^*$ only on a set of times of zero length earns
the same profit rate.

Finally, $R(X,t)$ is non-decreasing in $t$, so the set of times at which
$w^{(j)}R(X,t)\ge c_{\mathrm{op}}(i)$ is either empty or an interval $[t^*,\infty)$.
Under $a_i^*$, participant $i$ is therefore active from $t^*$ until the round ends.
\Halmos
\endproof

\section{The Backlog and the Mining Gap}
\label{ec:backlog}
\begin{lemma}
\label{lem:root}
Suppose $\theta s_b>1$. Then \eqref{eq:char} has a unique positive root $K\in(0,\theta)$,
given by $K=\theta+W_0\bigl(-\theta s_b e^{-\theta s_b}\bigr)/s_b$. At this root
$\beta\in(0,1)$, and
\begin{equation}
\frac{dK}{d\theta}=\frac{K}{\theta\beta}>0.
\label{eq:dKdtheta}
\end{equation}
Moreover, $K\to0$ as $\theta s_b\downarrow1$, and $K\to\infty$ with $K/\theta\to1$ as
$\theta\to\infty$.
\end{lemma}
\proof{Proof.}
Let $h(K):=\theta(1-e^{-Ks_b})$. The function $h$ is strictly increasing and strictly
concave, with $h(0)=0$, $h'(0)=\theta s_b>1$ and $h<\theta$. Hence $h(K)>K$ for small
$K>0$ and $h(K)<K$ for $K\ge\theta$, and by strict concavity the equation $h(K)=K$ has
exactly one positive root, which lies in $(0,\theta)$. At this root the concave function
$h$ crosses the identity from above, so $0<h'(K)=\theta s_b e^{-Ks_b}<1$, which gives
$\beta\in(0,1)$.

For the closed form, set $a:=\theta s_b>1$ and $u:=Ks_b$. Then \eqref{eq:char} reads
$u=a(1-e^{-u})$, which is equivalent to $(u-a)e^{u-a}=-ae^{-a}$. The right side lies in
$(-1/e,0)$, where the Lambert W function has two real branches. The lower branch returns
$u-a=-a$, which is the trivial root $u=0$. The principal branch returns
$u-a=W_0(-ae^{-a})>-1>-a$, so $u>0$, and this is the positive root.

Differentiating $K=h(K)$ with respect to $\theta$ gives $K'=(1-e^{-Ks_b})+h'(K)K'$.
Since $1-e^{-Ks_b}=K/\theta$, this yields $K'=K/\bigl(\theta(1-h'(K))\bigr)=K/(\theta\beta)$.

For the limits, fix $k>0$. Because $(1-e^{-x})/x<1$ for $x>0$, we have $h(k)<k$ once
$\theta s_b$ is close enough to one, so the positive root lies below $k$. Since $k$ is
arbitrary, $K\to0$. As $\theta\to\infty$, fix $\theta_0$ with $\theta_0s_b>1$
and let $K_0$ be the corresponding root. For $\theta\ge\theta_0$, $K\ge K_0$ because $K$
increases in $\theta$, so $K=h(K)\ge\theta(1-e^{-K_0s_b})$, which grows without bound.
Then $K/\theta=1-e^{-Ks_b}\to1$.
\Halmos
\endproof

The stationary law in \cref{prop:backlog} follows from the round-by-round distribution
of the backlog. Write $F_X^{(n)}(x):=\Pr(X^{(n)}\le x)$. There are two cases. Suppose
$X^{(n-1)}\le s_1$. Then $\max(X^{(n-1)},s_1)=s_1$, and the event $\{X^{(n)}\le x\}$
reduces to $\{\lambda T_\mu^{(n-1)}\le x+s_b-s_1\}$, which has probability
$1-e^{-\theta(x+s_b-s_1)}$. Suppose instead that $X^{(n-1)}=y$ for some $y>s_1$. Then
the event $\{X^{(n)}\le x\}$ becomes $\{\lambda T_\mu^{(n-1)}\le x+s_b-y\}$, with
probability $1-e^{-\theta(x+s_b-y)}$. The law of total probability gives
\begin{equation}
\begin{aligned}
F_X^{(n)}(x) =\; & F_X^{(n-1)}\bigl(s_1\bigr)\Bigl[1 -
e^{-\theta(x+s_b-s_1)}\Bigr]\\
&+ \int_{s_1}^{x+s_b}\Bigl[1 - e^{-\theta(x+s_b-y)}\Bigr]\,
dF_X^{(n-1)}(y)
\label{eq:transientcdf}
\end{aligned}
\end{equation}
The first term collects all rounds where the previous backlog was at or below $s_1$. In
those rounds, mining starts from exactly $s_1$, and the next-round backlog depends only
on the mining duration. The second term integrates over all previous backlogs above
$s_1$.

We pass \eqref{eq:transientcdf} to its stationary limit, in which
$F_X^{(n)}=F_X^{(n-1)}=F_X$. Since $\lambda T_\mu^{(n-1)}$ has a density, $X$ has no atom
on $(0,\infty)$, and the integral in \eqref{eq:transientcdf}, which runs over
$y\in(s_1,x+s_b]$, can be written with a density. Differentiating
\eqref{eq:transientcdf} once with respect to $x$ yields the steady-state integral
equation
\begin{equation}
\begin{aligned}
f_X(x)={}&\theta e^{-\theta(x+s_b)}\Bigl[F_X(s_1)\,e^{\theta s_1}\\
&+\int_{s_1}^{x+s_b}e^{\theta y}f_X(y)\,dy\Bigr],\qquad x>0.
\end{aligned}
\label{eq:integraleq}
\end{equation}
Differentiating once more eliminates the integral and produces the delay differential
equation
\begin{equation}
f_X'(x) = \theta\bigl[f_X(x+s_b) - f_X(x)\bigr], \qquad x > 0.
\label{eq:dde}
\end{equation}
Substituting the exponential ansatz $f_X(x)=Ae^{-Kx}$ into \eqref{eq:dde} gives
$-K=\theta\bigl(e^{-Ks_b}-1\bigr)$, which is the characteristic equation
\eqref{eq:char}. We take its positive root $K$ from \cref{lem:root}, so that the
density is integrable.

It remains to find the constant $A$. Equation \eqref{eq:dde} cannot determine it,
since any constant multiple of a solution of \eqref{eq:dde} is again a solution. The
information about $A$ was lost when we differentiated \eqref{eq:integraleq}, so we
return to \eqref{eq:integraleq}. Normalization gives $p_0=1-A/K$, and hence
$F_X(s_1)=p_0+\int_0^{s_1}Ae^{-Ky}\,dy=1-(A/K)e^{-Ks_1}$. Under the ansatz, the
integral in \eqref{eq:integraleq} equals
\[
\int_{s_1}^{x+s_b}e^{\theta y}Ae^{-Ky}\,dy
=\frac{A\bigl(e^{(\theta-K)(x+s_b)}-e^{(\theta-K)s_1}\bigr)}{\theta-K},
\]
so the right side of \eqref{eq:integraleq} becomes
\[
\begin{aligned}
&\frac{\theta A\,e^{-K(x+s_b)}}{\theta-K}\\
&+\theta e^{-\theta(x+s_b)}\Bigl[F_X(s_1)\,e^{\theta s_1}
-\frac{A\,e^{(\theta-K)s_1}}{\theta-K}\Bigr].
\end{aligned}
\]
By \eqref{eq:char}, $\theta-K=\theta e^{-Ks_b}$, so the first term equals $Ae^{-Kx}$,
which is the left side $f_X(x)$ of \eqref{eq:integraleq}. The second term is a
constant multiple of $e^{-\theta x}$. Equation \eqref{eq:integraleq} therefore holds
for every $x>0$ if and only if the constant in brackets is zero, that is,
\[
F_X(s_1)\,e^{\theta s_1}=\frac{A\,e^{(\theta-K)s_1}}{\theta-K}.
\]
Substituting $F_X(s_1)=1-(A/K)e^{-Ks_1}$ and $\theta-K=\theta e^{-Ks_b}$ and solving
for $A$ gives $A=Ke^{-K(s_b-s_1)}$. This is the law in \eqref{eq:stationarybacklog}. The proof below shows that it is the only stationary law
and that the backlog converges to it from every initial condition.

\proof{Proof of \cref{prop:backlog}.}
Write $Z_n:=\lambda T_\mu^{(n)}$. These are independent and exponentially distributed with
rate $\theta$. Let $Y^{(n)}:=\max(X^{(n)},s_1)-s_1\ge0$ be the backlog in excess of the
threshold. Since $s_1\ge0$, we have $\max(x^+,s_1)=\max(x,s_1)$ for every real $x$, and
\eqref{eq:recursion} gives
\begin{equation}
\begin{aligned}
Y^{(n+1)}&=\bigl(Y^{(n)}+Z_n-s_b\bigr)^+,\\
X^{(n+1)}&=\bigl(s_1+Y^{(n)}+Z_n-s_b\bigr)^+.
\end{aligned}
\label{eq:lindleyY}
\end{equation}
The first equation is the Lindley recursion for the waiting time in a single-server queue
with deterministic interarrival times $s_b$ and exponential service times $Z_n$. Its
increments $Z_n-s_b$ are independent with mean $1/\theta-s_b<0$. The recursion therefore
has a unique stationary distribution, and $Y^{(n)}$ converges to it in distribution from
every initial condition \citep{asmussen2003applied}. Since $Z_n$ is independent of
$Y^{(n)}$ and $X^{(n+1)}$ is a continuous function of $(Y^{(n)},Z_n)$, the backlog $X^{(n)}$ converges in distribution as well.

It remains to show that the law derived above is this stationary law. Let $K$ be the root in \cref{lem:root} and
$\sigma:=e^{-Ks_b}\in(0,1)$. For $y>0$, the law derived above gives
$\Pr(Y>y)=\Pr(X>s_1+y)=e^{-K(s_b+y)}=\sigma e^{-Ky}$. Suppose that $Y$ has this law, so
that $Y$ has the atom $1-\sigma$ at zero and the density $\sigma Ke^{-Ky}$ on $(0,\infty)$, and let
$Z$ be independent of $Y$ and exponential with rate $\theta$. For $u>0$,
\[
\begin{aligned}
\Pr(Y+Z>u)={}&(1-\sigma)e^{-\theta u}+\sigma e^{-Ku}\\
&+\int_0^u e^{-\theta(u-y)}\,\sigma Ke^{-Ky}\,dy.
\end{aligned}
\]
By \eqref{eq:char}, $K=\theta(1-\sigma)$ and $\theta-K=\theta\sigma$, so the integral
equals $\sigma K(e^{-Ku}-e^{-\theta u})/(\theta-K)=(1-\sigma)(e^{-Ku}-e^{-\theta u})$.
Hence
\begin{equation}
\Pr(Y+Z>u)=e^{-Ku},\qquad u>0.
\label{eq:YZtail}
\end{equation}
Taking $u=y+s_b$ in the first equation of \eqref{eq:lindleyY} gives
$\Pr(Y^{(n+1)}>y)=e^{-K(y+s_b)}=\sigma e^{-Ky}$ for $y>0$. The assumed law is therefore
invariant, and by uniqueness it is the stationary law of $Y$. Taking $u=x+s_b-s_1$, which
is positive because $s_1<s_b$, in the second equation gives, under stationarity,
\[
\Pr(X>x)=e^{-K(s_b-s_1+x)},\qquad x\ge0.
\]
Its value at $x=0$ gives $p_0=1-e^{-K(s_b-s_1)}$, and differentiating in $x$ gives the
density in \eqref{eq:stationarybacklog}.
\Halmos
\endproof

\proof{Proof of \cref{cor:gap}.}
By \eqref{eq:gapdef}, $T_g=0$ if $X\ge s_1$, $T_g=s_1/\lambda$ if $X=0$, and
$T_g=(s_1-X)/\lambda$ otherwise. For $0\le t<s_1/\lambda$, the event $\{T_g\le t\}$
equals $\{X\ge s_1-\lambda t\}$. Since $s_1-\lambda t>0$ and $X$ has no atom on
$(0,\infty)$, \cref{prop:backlog} gives $\Pr(T_g\le t)=e^{-K(s_b-\lambda t)}$, which is
\eqref{eq:gapcdf}. The atom at zero is $\Pr(X\ge s_1)=e^{-Ks_b}$, and the atom at
$s_1/\lambda$ is $p_0$. Integrating the survival function gives
\[
\begin{aligned}
\mathbb{E}[T_g]&=\int_0^{s_1/\lambda}\bigl(1-e^{-K(s_b-\lambda t)}\bigr)\,dt\\
&=\frac{s_1}{\lambda}-\frac{e^{-Ks_b}}{K\lambda}\bigl(e^{Ks_1}-1\bigr),
\end{aligned}
\]
which is \eqref{eq:expectedgap}.
\Halmos
\endproof

\section{The Difficulty Adjustment Map and Stability}
\label{app:stability}
Throughout this section $\gamma>0$ is fixed, and the winning probability ranges over
$\mathcal{D}_\gamma:=(\underline w_\gamma,\infty)$, where
$\underline w_\gamma:=\max\bigl(\underline w,\lambda/(\gamma s_b)\bigr)$ and
$\underline w=\bar c_{\mathrm{op}}/R_{\max}$. On $\mathcal D_\gamma$ the backlog is
stable, $\theta s_b>1$, and some mining is profitable, $s_1(w)<s_b$, so
\cref{prop:backlog} applies.

\proof{Proof of \cref{prop:invariant}.}
By \eqref{eq:mdef} and \cref{cor:gap}, $m(w)=\mathbb{E}[T_g]+1/\mu$. The characteristic
equation gives $e^{-Ks_b}=1-K/\theta=1-K\lambda/\mu$, so
$e^{-Ks_b}/(K\lambda)=1/(K\lambda)-1/\mu$. Substituting \eqref{eq:expectedgap} into
$m(w)=\mathbb{E}[T_g]+1/\mu$ and using this identity for the term $e^{-Ks_b}/(K\lambda)$
gives \eqref{eq:mw}, which in the notation of \cref{ec:notation} reads
\begin{equation}
\lambda\,m(w)=s_1(w)+g\bigl(K(w),\phi(w)\bigr),
\label{eq:mg}
\end{equation}
with $\phi(w)=s_b-s_1(w)$. The threshold $s_1(w)$ is continuous. It equals $(\bar c_{\mathrm{op}}/w-R_0)/\alpha$ for
$w<\bar c_{\mathrm{op}}/R_0$ and zero for $w\ge\bar c_{\mathrm{op}}/R_0$. Since
$\theta=w\gamma/\lambda$, \cref{lem:root} shows that $K(w)$ is continuous and strictly
increasing. Hence $m$ is continuous, and on each side of $\bar c_{\mathrm{op}}/R_0$,
\begin{equation}
\lambda\,m'(w)=s_1'(w)\,(1-g_\phi)+g_K\,K'(w),
\label{eq:mprime}
\end{equation}
where $K'(w)=K/(w\beta)$ by \eqref{eq:dKdtheta} and $d\theta/dw=\theta/w$. The first term
is non-positive because $s_1'\le0$ and $g_\phi<1$. The second term is negative because
$g_K<0$ and $K'>0$. Hence $m$ is strictly decreasing on $\mathcal D_\gamma$.

For $w\ge\bar c_{\mathrm{op}}/R_0$, we have $s_1=0$ and
$g(K,s_b)=(1-e^{-Ks_b})/K=1/\theta$ by \eqref{eq:char}, so $m(w)=1/(w\gamma)$, which
tends to zero as $w\to\infty$. At the lower end of $\mathcal D_\gamma$ there are two
cases. If $\underline w_\gamma=\underline w$, then $\phi\to0$, and $g(K,\phi)\le\phi$
tends to zero. Otherwise $\theta s_b\downarrow1$, so $K\to0$ by \cref{lem:root} and
$g(K,\phi)\to\phi$. In both cases $\lambda m(w)\to s_1+\phi=s_b$. Since
$0<B<s_b/\lambda$, continuity and strict monotonicity give a unique
$w^*\in\mathcal D_\gamma$ with $m(w^*)=B$.
\Halmos
\endproof

\proof{Proof of \cref{prop:stability}.}
Away from $\bar c_{\mathrm{op}}/R_0$, the function $m$ is continuously differentiable,
and so is the map $\Phi(w)=w\,m(w)/B$ of \eqref{eq:phimap}. At the fixed point,
$m(w^*)=B$ gives $\Phi'(w^*)=1+w^*m'(w^*)/B=1+\varepsilon(w^*)$. A fixed point of a
continuously differentiable map is locally asymptotically stable if
$|\Phi'(w^*)|<1$ and unstable if $|\Phi'(w^*)|>1$. Since $\varepsilon(w^*)<0$ by
\cref{prop:invariant}, the first case is $-2<\varepsilon(w^*)<0$ and the second is
$\varepsilon(w^*)<-2$.

Without a gap, $m(w)=1/(w\gamma)$ near $w^*$, so $\varepsilon(w^*)=-1$.

With a gap, $s_1^*>0$ and $s_1'(w^*)=-\bar c_{\mathrm{op}}/(\alpha w^{*2})$. The
threshold \eqref{eq:s1common} gives $\bar c_{\mathrm{op}}/(\alpha w^*)=s_1^*+V$, so
$w^*s_1'(w^*)=-(s_1^*+V)$. The condition $m(w^*)=B$ and \eqref{eq:mg} give
$g(K^*,\phi^*)=Q^*$, that is, $1-e^{-K^*\phi^*}=K^*Q^*$. Substituting into
\eqref{eq:mprime} yields
\begin{equation}
\begin{aligned}
\varepsilon(w^*)={}&-\frac{K^*(s_1^*+V)\,Q^*}{\lambda B}\\
&-\frac{K^*|g_K^*|}{\lambda B\,\beta^*},
\end{aligned}
\label{eq:epsgap}
\end{equation}
where $|g_K^*|=(Q^*-\phi^*e^{-K^*\phi^*})/K^*$. The invariant winning probability is
therefore stable if the right side of \eqref{eq:epsgap} exceeds $-2$ and unstable if it
falls below $-2$. The first term is a threshold effect, through which $w$ moves the gap by moving $s_1$. The second
is a backlog effect, through which $w$ reshapes the stationary backlog distribution.

For the limit, let $\gamma\to\infty$. By \cref{prop:hashpower}, the invariant point has a
gap for $\gamma>\gamma_0$, and $\mu^*\to\infty$ and $s_1^*\to\lambda B$. Hence
$Q^*\to0$ and $\phi^*\to D>0$. By \cref{lem:root}, $K^*\to\infty$ and $K^*\ge\theta/2$
for large $\theta$, so $\theta s_be^{-K^*s_b}\to0$ and $\beta^*\to1$. The identity
$K^*Q^*=1-e^{-K^*\phi^*}$ gives $K^*Q^*\to1$, and $0<K^*|g_K^*|\le Q^*\to0$.
Substituting into \eqref{eq:epsgap},
\begin{equation}
\lim_{\gamma\to\infty}\varepsilon(w^*)=-\frac{\lambda B+V}{\lambda B}
=-1-\frac{R_b+s_br_0}{(r_1-r_0)\lambda B}.
\label{eq:epslimit}
\end{equation}
This limit exceeds $-2$ if and only if $R_b+s_br_0<(r_1-r_0)\lambda B$. It equals the
elasticity that \citet[Corollary~3]{kimSeo2026mind} obtain for identical miners with an
affine reward and a mempool that resets after every block. Our backlog approaches such a
mempool in the limit, since $p_0=1-e^{-K^*\phi^*}\to1$.
\Halmos
\endproof

The next lemma gives the elasticity when the Type-1 backlog is empty at the start of every
round, as in a mempool that resets after every block. It holds at every aggregate
hashpower, not only in the limit of \cref{prop:stability}, and it matches the elasticity
that \citet[Corollary~3]{kimSeo2026mind} compute for finitely many identical miners.
\begin{lemma}
\label{lem:reset}
Suppose that every round starts with no Type-1 transactions waiting and that the invariant
winning probability $w^*$ has a gap, $s_1(w^*)>0$. Then
$\varepsilon(w^*)=-1-\nu$, where $\nu=(R_b+s_br_0)/((r_1-r_0)\lambda B)$.
\end{lemma}
\proof{Proof.}
With an empty backlog at the start of every round, the gap lasts $s_1(w)/\lambda$, and the
active phase that follows lasts $1/(w\gamma)$ on average, so
\[
m(w)=\frac{s_1(w)}{\lambda}+\frac{1}{w\gamma}.
\]
Near $w^*$, the threshold is positive and equals $(\bar c_{\mathrm{op}}/w-R_0)/\alpha$ by
\eqref{eq:s1common}, so $w\,s_1'(w)=-\bar c_{\mathrm{op}}/(\alpha w)=-(s_1(w)+V)$ with
$V=R_0/\alpha$. Hence
\[
w\,m'(w)=-\frac{s_1(w)+V}{\lambda}-\frac{1}{w\gamma}=-m(w)-\frac{V}{\lambda}.
\]
At $w^*$ we have $m(w^*)=B$, so
$\varepsilon(w^*)=w^*m'(w^*)/m(w^*)=-1-V/(\lambda B)=-1-\nu$.
\Halmos
\endproof

\section{Aggregate Hashpower}
\label{ec:hashpower}
\begin{lemma}
\label{lem:daacurve}
Suppose $\lambda B<s_b$. For every $\mu>1/B$, equation \eqref{eq:sysA2}, with $K$ the root
of \eqref{eq:char} at $\theta=\mu/\lambda$, has a unique solution
$s_1=\hat s_1(\mu)\in(0,\lambda B)$. The function $\hat s_1$ is continuously
differentiable and strictly increasing, with
\begin{equation}
\zeta(\mu):=\hat s_1'(\mu)=-\frac{g_K\,K}{\mu\beta\,\bigl(1-e^{-K\phi}\bigr)}>0.
\label{eq:zeta}
\end{equation}
Moreover, $\hat s_1(\mu)\to0$ as $\mu\downarrow1/B$ and $\hat s_1(\mu)\to\lambda B$ as
$\mu\to\infty$.
\end{lemma}
\proof{Proof.}
For $\mu>1/B$, $\theta s_b=\mu s_b/\lambda>s_b/(\lambda B)>1$, so $K$ is well defined.
In the notation of \cref{ec:notation}, \eqref{eq:sysA2} reads $\psi(s_1)=\lambda B$ with
$\psi(s_1):=s_1+g(K,s_b-s_1)$. On $[0,\lambda B]$, $\psi$ is continuous with
$\psi'=1-g_\phi\ge1-e^{-KD}>0$. It satisfies $\psi(0)=g(K,s_b)=1/\theta=\lambda/\mu<\lambda B$
and $\psi(\lambda B)>\lambda B$, so the root is unique and lies in $(0,\lambda B)$. The
implicit function theorem, with $\partial\psi/\partial s_1=1-g_\phi>0$ and
$dK/d\mu=K/(\mu\beta)$ from \eqref{eq:dKdtheta}, gives \eqref{eq:zeta}, which is
positive because $g_K<0$. As $\mu\downarrow1/B$, $\psi(0)\uparrow\lambda B$, and since
$\psi'$ is bounded below by a positive constant, the root tends to zero. As
$\mu\to\infty$, $K\to\infty$ and $0<g(K,\phi)\le1/K\to0$, so the root tends to
$\lambda B$.
\Halmos
\endproof

\proof{Proof of \cref{prop:hashpower}.}
(i) Suppose $\gamma\le\gamma_0$ and let $w=1/(\gamma B)$. Then
$\bar c_{\mathrm{op}}/w=\bar c_{\mathrm{op}}\gamma B\le R_0$, so $s_1(w)=0$ and
$m(w)=1/(w\gamma)=B$. This $w$ lies in $\mathcal D_\gamma$ because
$\theta s_b=s_b/(\lambda B)>1$ and $s_1(w)<s_b$. By the uniqueness in
\cref{prop:invariant}, $w^*=1/(\gamma B)$ and there is no gap. Conversely, if
$s_1(w^*)=0$, then $m(w^*)=1/(w^*\gamma)=B$, so $w^*=1/(\gamma B)$, and $s_1(w^*)=0$
requires $\bar c_{\mathrm{op}}\gamma B\le R_0$, that is, $\gamma\le\gamma_0$.

(ii) Let $\gamma>\gamma_0$, so that $s_1^*>0$. Then $p_0>0$, so $\mathbb{E}[T_g]>0$, and
$m(w^*)=\mathbb{E}[T_g]+1/\mu^*=B$ gives $\mu^*>1/B$. By \cref{lem:daacurve},
$s_1^*=\hat s_1(\mu^*)$, and the threshold \eqref{eq:sysB2} gives
$\bar c_{\mathrm{op}}\gamma=H(\mu^*)$ with
\[
H(\mu):=\mu\bigl(R_0+\alpha\hat s_1(\mu)\bigr),\qquad \mu>1/B.
\]
The function $H$ is continuous and strictly increasing, with
$H(\mu)\to R_0/B=\bar c_{\mathrm{op}}\gamma_0$ as $\mu\downarrow1/B$ and
$H(\mu)\to\infty$ as $\mu\to\infty$. Hence $\mu^*=H^{-1}(\bar c_{\mathrm{op}}\gamma)$ is
continuous and strictly increasing in $\gamma>\gamma_0$, and so are
$s_1^*=\hat s_1(\mu^*)$, $\mathbb{E}[T_g]=B-1/\mu^*$ and
$\mathcal W=\gamma/\mu^*=(R_0+\alpha s_1^*)/\bar c_{\mathrm{op}}$.

(iii) Since $\hat s_1<\lambda B$,
$\bar c_{\mathrm{op}}\gamma=H(\mu^*)<\mu^*(R_0+\alpha\lambda B)=\mu^*\bar c_{\mathrm{op}}C$,
so $\mu^*>\gamma/C\to\infty$. \Cref{lem:daacurve} then gives $s_1^*\to\lambda B$, so
$\mathbb{E}[T_g]\to B$ and $\mathcal W\to(R_0+\alpha\lambda B)/\bar c_{\mathrm{op}}=C$,
from below because $s_1^*<\lambda B$.
\Halmos
\endproof

\section{Free Entry and the No-Gap Condition}
\label{ec:entry}
\proof{Proof of \cref{prop:entry}.}
Let $P(\gamma)$ denote the expected profit per unit of hashpower, the left side of
\eqref{eq:freeentry} minus its right side. For $\gamma\le\gamma_0$,
\cref{prop:hashpower}(i) gives $\mu^*=1/B$, so
$P(\gamma)=(R_0+\alpha\lambda B)/\gamma-(\bar c_{\mathrm{cap}}+\bar c_{\mathrm{op}})B$.
For $\gamma>\gamma_0$, the threshold \eqref{eq:sysB2} gives
$\bar c_{\mathrm{op}}/\mu^*=(R_0+\alpha s_1^*)/\gamma$, so
\begin{equation}
P(\gamma)=\frac{\alpha Q^*(\gamma)}{\gamma}-\bar c_{\mathrm{cap}}B,
\label{eq:profitgap}
\end{equation}
where $Q^*(\gamma)=\lambda B-s_1^*\in(0,\lambda B)$. Both expressions equal $\alpha\lambda B/\gamma_0-\bar c_{\mathrm{cap}}B$ at
$\gamma=\gamma_0$, because $R_0/\gamma_0=\bar c_{\mathrm{op}}B$ and $s_1^*\to0$ as
$\gamma\downarrow\gamma_0$. Hence $P$ is continuous. It is strictly decreasing on both
pieces, since $Q^*$ is strictly decreasing in $\gamma$ by \cref{prop:hashpower}(ii). As
$\gamma\downarrow0$, $P(\gamma)\to\infty$. As $\gamma\to\infty$,
$0<\alpha Q^*/\gamma<\alpha\lambda B/\gamma\to0$, so
$P(\gamma)\to-\bar c_{\mathrm{cap}}B<0$. Hence $P$ has a unique zero $\gamma^*$, with
$P>0$ below $\gamma^*$ and $P<0$ above it. If $\gamma^*>\gamma_0$, the equation
$P(\gamma^*)=0$ is \eqref{eq:entrygap}. If the capital cost were zero, then
$P(\gamma)=\alpha Q^*/\gamma>0$ for $\gamma>\gamma_0$ and
$P(\gamma)\ge P(\gamma_0)=\alpha\lambda B/\gamma_0>0$ for $\gamma\le\gamma_0$, so $P$
would have no zero.
\Halmos
\endproof

\proof{Proof of \cref{prop:nogap}.}
Since $P$ is strictly decreasing, $\gamma^*\le\gamma_0$ if and only if
$P(\gamma_0)\le0$, that is, $\alpha\lambda B/\gamma_0\le\bar c_{\mathrm{cap}}B$. With
$\gamma_0=R_0/(\bar c_{\mathrm{op}}B)$ from \eqref{eq:gamma0}, this is
$\bar c_{\mathrm{op}}/\bar c_{\mathrm{cap}}\le R_0/(\alpha\lambda B)$, which is
\eqref{eq:nogapcondition}. In this case $\mu^*=1/B$, and $P(\gamma^*)=0$ is
\eqref{eq:nogapentry}, so
$\gamma^*=(R_0+\alpha\lambda B)/\bigl((\bar c_{\mathrm{cap}}+\bar c_{\mathrm{op}})B\bigr)$.

The signs in \cref{tab:prevent} follow by inspection. A higher $R_b$, or a higher $r_0$
at fixed $r_1$, raises $R_0/(\alpha\lambda B)$, because $r_0$ raises $R_0$ and lowers
$\alpha$. A higher $\alpha$ lowers it. A higher $\bar c_{\mathrm{cap}}$ lowers
$\bar c_{\mathrm{op}}/\bar c_{\mathrm{cap}}$, and a higher $\bar c_{\mathrm{op}}$ raises
it. For $\gamma^*$, the derivatives with respect to $R_b$ and $\alpha$ are positive, those
with respect to $\bar c_{\mathrm{cap}}$ and $\bar c_{\mathrm{op}}$ are negative, and at
fixed $r_1$ the derivative of $R_0+\alpha\lambda B$ with respect to $r_0$ is
$s_b-\lambda B=D>0$.
\Halmos
\endproof

\section{Parameter Effects in the Gap Regime}
\label{app:statics}
Throughout this section $\lambda B<s_b$, and the free-entry equilibrium has a gap,
$\gamma^*>\gamma_0$. By the proof of \cref{prop:hashpower}, the equilibrium satisfies
$s_1^*=\hat s_1(\mu^*)$, the threshold
$\bar c_{\mathrm{op}}\gamma^*=\mu^*(R_0+\alpha s_1^*)$, and the free-entry condition
$\alpha Q^*=\bar c_{\mathrm{cap}}B\gamma^*$ of \eqref{eq:entrygap}. Eliminating
$\gamma^*$ and dividing by $\alpha\mu^*$ gives the master equation
$\Psi(\mu^*;V,\Gamma)=0$, with
\begin{equation}
\Psi(\mu;V,\Gamma):=\frac{\Gamma\hat Q(\mu)}{B\mu}-\hat s_1(\mu)-V,
\label{eq:master}
\end{equation}
where $\hat Q(\mu):=\lambda B-\hat s_1(\mu)$. The functions $\hat s_1$ and $\hat Q$ depend
only on $\lambda$, $s_b$ and $B$, so the fee and cost parameters enter
\eqref{eq:master} only through $V=R_0/\alpha$ and
$\Gamma=\bar c_{\mathrm{op}}/\bar c_{\mathrm{cap}}$. By \cref{lem:daacurve},
$d\hat s_1/d\mu=\zeta$ and $d\hat Q/d\mu=-\zeta$, so
\[
\frac{\partial\Psi}{\partial\mu}=-\frac{\Gamma\zeta}{B\mu}-\frac{\Gamma\hat Q}{B\mu^2}-\zeta
=-\frac{\Omega}{B\mu},
\]
where $\Omega:=(B\mu+\Gamma)\zeta+\Gamma\hat Q/\mu>0$.
At the equilibrium, \eqref{eq:master} gives $\Gamma Q^*/\mu^*=(s_1^*+V)B$, so
$\Omega=(B\mu^*+\Gamma)\zeta+(s_1^*+V)B$. Since $\partial\Psi/\partial\mu<0$, the master
equation has at most one root, and the implicit function theorem gives
\begin{equation}
\frac{d\mu^*}{dV}=-\frac{B\mu^*}{\Omega},\qquad
\frac{d\mu^*}{d\Gamma}=\frac{Q^*}{\Omega},
\label{eq:pert}
\end{equation}
and $ds_1^*=\zeta\,d\mu^*$ and $dQ^*=-\zeta\,d\mu^*$ along the curve of
\cref{lem:daacurve}. By \cref{prop:hashpower}(ii) and \eqref{eq:entrygap}, the three
metrics are
\begin{equation}
\begin{aligned}
\mathbb{E}[T_g]&=B-\frac{1}{\mu^*},\\
\gamma^*&=\frac{\alpha Q^*}{\bar c_{\mathrm{cap}}B},\\
\mathcal W&=\frac{R_0+\alpha s_1^*}{\bar c_{\mathrm{op}}},
\end{aligned}
\label{eq:metrics}
\end{equation}
so a parameter shortens the gap if and only if it lowers $\mu^*$. Two identities
simplify the algebra below. Expanding $\Omega$ and using
$(s_1^*+V)B=\Gamma Q^*/\mu^*$ gives $\Omega-B\mu^*\zeta=\Gamma\zeta+(s_1^*+V)B>0$, and
$\phi^*-Q^*=D>0$.

\proof{Proof of \cref{prop:statics}.}
\textit{Block subsidy.} $R_b$ enters only through $V$, with $\partial V/\partial R_b=1/\alpha$.
By \eqref{eq:pert}, $d\mu^*/dR_b=-B\mu^*/(\alpha\Omega)<0$, so the gap shortens. By
\eqref{eq:metrics},
\[
\begin{aligned}
\bar c_{\mathrm{op}}\frac{d\mathcal W}{dR_b}&=1+\alpha\zeta\frac{d\mu^*}{dR_b}
=\frac{\Omega-B\mu^*\zeta}{\Omega}>0,\\
\frac{d\gamma^*}{dR_b}&=-\frac{\alpha\zeta}{\bar c_{\mathrm{cap}}B}\frac{d\mu^*}{dR_b}
=\frac{\mu^*\zeta}{\bar c_{\mathrm{cap}}\Omega}>0.
\end{aligned}
\]

\textit{Low-tier fee.} With $r_1$ fixed, $\partial R_0/\partial r_0=s_b$ and
$\partial\alpha/\partial r_0=-1$, so $\partial V/\partial r_0=(s_b+V)/\alpha>0$ and
$d\mu^*/dr_0=-B\mu^*(s_b+V)/(\alpha\Omega)<0$. The gap shortens. For the hashwork,
\[
\begin{aligned}
\bar c_{\mathrm{op}}\frac{d\mathcal W}{dr_0}&=\phi^*+\alpha\zeta\frac{d\mu^*}{dr_0}
=\frac{\phi^*\Omega-B\mu^*\zeta(s_b+V)}{\Omega}\\
&=\frac{\Gamma\bigl(\zeta D+\phi^*Q^*/\mu^*\bigr)}{\Omega}>0,
\end{aligned}
\]
where the last step expands $\Omega$, substitutes $(s_1^*+V)B=\Gamma Q^*/\mu^*$ and
$s_b+V=\phi^*+s_1^*+V$, and uses $\phi^*-Q^*=D$. For the hashpower,
\[
\frac{d\gamma^*}{dr_0}=\frac{-Q^*-\alpha\zeta\,d\mu^*/dr_0}{\bar c_{\mathrm{cap}}B}
=\frac{B\mu^*\zeta D-\Gamma Q^{*2}/\mu^*}{\bar c_{\mathrm{cap}}B\,\Omega},
\]
which is positive if $\Gamma<\Gamma^\star$ and negative if $\Gamma>\Gamma^\star$, with
$\Gamma^\star:=B\mu^{*2}\zeta D/Q^{*2}$. Since $\mu^*$ depends on $\Gamma$, the threshold
$\Gamma^\star$ is evaluated at the equilibrium and is not a fixed number.

\textit{Fee premium.} With $r_0$ fixed, $R_0$ is fixed and
$\partial V/\partial\alpha=-V/\alpha$, so $d\mu^*/d\alpha=B\mu^*V/(\alpha\Omega)>0$. The
gap lengthens. Then
\[
\begin{aligned}
\bar c_{\mathrm{op}}\frac{d\mathcal W}{d\alpha}&=s_1^*+\alpha\zeta\frac{d\mu^*}{d\alpha}>0,\\
\frac{d\gamma^*}{d\alpha}&=\frac{Q^*\Omega-B\mu^*\zeta V}{\bar c_{\mathrm{cap}}B\,\Omega}\\
&=\frac{\lambda B^2\mu^*\zeta+\Gamma Q^{*2}/\mu^*}{\bar c_{\mathrm{cap}}B\,\Omega}>0,
\end{aligned}
\]
where the last step expands $\Omega$, substitutes $V=\Gamma Q^*/(B\mu^*)-s_1^*$, and uses
$Q^*+s_1^*=\lambda B$.

\textit{Capital cost.} $\bar c_{\mathrm{cap}}$ enters only through $\Gamma$, with
$d\Gamma/d\bar c_{\mathrm{cap}}=-\Gamma/\bar c_{\mathrm{cap}}$, so
$d\mu^*/d\bar c_{\mathrm{cap}}=-\Gamma Q^*/(\bar c_{\mathrm{cap}}\Omega)<0$. The gap
shortens. Since $\mathcal W$ does not depend on $\bar c_{\mathrm{cap}}$ directly,
\[
\begin{aligned}
\frac{d\mathcal W}{d\bar c_{\mathrm{cap}}}&=\frac{\alpha\zeta}{\bar c_{\mathrm{op}}}
\frac{d\mu^*}{d\bar c_{\mathrm{cap}}}<0,\\
\frac{d\gamma^*}{d\bar c_{\mathrm{cap}}}
&=\frac{\alpha Q^*}{\bar c_{\mathrm{cap}}^2B}\left(\frac{\Gamma\zeta}{\Omega}-1\right)<0,
\end{aligned}
\]
because $\Omega>\Gamma\zeta$.

\textit{Operating cost.} $\bar c_{\mathrm{op}}$ enters only through $\Gamma$, with
$d\Gamma/d\bar c_{\mathrm{op}}=1/\bar c_{\mathrm{cap}}$, so
$d\mu^*/d\bar c_{\mathrm{op}}=Q^*/(\bar c_{\mathrm{cap}}\Omega)>0$. The gap lengthens.
Then
\[
\begin{aligned}
\frac{d\gamma^*}{d\bar c_{\mathrm{op}}}
&=-\frac{\alpha\zeta Q^*}{\bar c_{\mathrm{cap}}^2B\,\Omega}<0,\\
\frac{d\mathcal W}{d\bar c_{\mathrm{op}}}
&=-\frac{\alpha Q^*\,(\Omega-B\mu^*\zeta)}{\bar c_{\mathrm{op}}\bar c_{\mathrm{cap}}B\mu^*\,\Omega}<0,
\end{aligned}
\]
where the second expression differentiates $\mathcal W=\alpha(V+s_1^*)/\bar c_{\mathrm{op}}$
and uses $(V+s_1^*)/\bar c_{\mathrm{op}}=Q^*/(\bar c_{\mathrm{cap}}B\mu^*)$.

Collecting these signs gives \cref{tab:reduce}.
\Halmos
\endproof

\section{Estimating the Fee Inputs}
\label{ec:fees}
This section explains how \cref{sec:numerical-estimation} estimates $s_b r_0$, the fee
premium $(r_1-r_0)\lambda B$ and $\lambda$ from block data. The data are the 12{,}096 blocks
between heights 957{,}600 and 969{,}695 \citep{mempoolspace2026}. For each block we use its
timestamp and total fees and, where stated, the size and the fee of each of its
transactions.

\paragraph{Fee accumulation.} By \eqref{eq:reward}, a block found $\tau$ minutes after its
predecessor earns the fees $s_b r_0+(r_1-r_0)(X+\lambda\tau)$ before it fills with Type-1
transactions, where $X$ is the Type-1 backlog at the start of its round. When the
predecessor was found quickly, it could include every Type-1 transaction waiting at that
time, so that $X$ is most likely zero and the fees are $s_b r_0+(r_1-r_0)\lambda\tau$. We
therefore keep the blocks whose predecessor was found within 2 minutes and whose own
$\tau$ is at most 5 minutes, 737 blocks in total. We limit $\tau$ to 5 minutes because the fees grow linearly in $\tau$ only
until the block fills with Type-1 transactions, and short intervals stay well below that
point. Fees per block are skewed to the right, so we fit a robust line, whose slope is the
median of the slopes between all pairs of blocks and whose intercept is the median fee
minus this slope times the median $\tau$. Panel (a) of \cref{fig:ec-fees} shows the
blocks, the median fee in each minute and the fitted line. Its intercept is
$s_b r_0=0.0045$~BTC, that is, $r_0=0.45$ satoshi per vB, and its slope gives
$(r_1-r_0)\lambda B=0.0154$~BTC per target round. Block timestamps are set by the miners and
are not exact, which adds noise to $\tau$.

\paragraph{The Type-1 arrival rate.} The block fees reveal only the product
$(r_1-r_0)\lambda$ and not its two factors separately. To separate $\lambda$, we use the size and the fee of every transaction
in the same 737 blocks. A transaction counts as Type-1 if its feerate exceeds the barrier
$r_0(1+p)$, and $\lambda$ is the slope of the total size of the Type-1 transactions of a block
on $\tau$, fitted in the same robust way. Panel (b) of \cref{fig:ec-fees} shows $\lambda$ for
margins $p$ from 5\% to 300\%. Up to $p=100\%$, $\lambda$ stays between 70{,}000 and
75{,}000~vB per minute. Beyond the point at which the barrier reaches 1 satoshi per vB,
$\lambda$ drops, because many transactions pay exactly this feerate, and their size grows with
$\tau$ like that of Type-1 transactions. We use $p=20\%$, which gives $\lambda=73{,}000$~vB per
minute and $\xi=\lambda B/s_b=0.73$.

\paragraph{Checks.} In the two-tier model, each arriving Type-1 transaction displaces a
Type-0 transaction that pays $r_0$, so the fees paid below the barrier fall at the rate
$r_0\lambda$. For $p=20\%$ they fall by 0.0031~BTC per target round, against 0.0032~BTC
predicted by $r_0\lambda B$. The same holds for every $p$ up to 100\%, and it fails above 1
satoshi per vB. A second check uses the identity that the fees of a block equal
$s_b r_0+(r_1-r_0)U$, where $U$ is its Type-1 size. A robust regression of the block fees
on $U$ does not depend on whether the Type-1 transactions arrived during the round or were
left over from earlier rounds. It gives $(r_1-r_0)$ between 1.7 and 1.9 satoshi per vB and a
value of $\nu$ without the subsidy between 0.29 and 0.30, in line with the estimate of 0.29.
Blocks found less than a minute after a fast predecessor still contain about 65{,}000~vB of
Type-1 transactions in the median, so some Type-1 transactions do carry over. On average,
these blocks contain about 160{,}000~vB of Type-1 transactions. The two-tier model, which
treats all Type-1 transactions alike, predicts about 225{,}000~vB, that is, more carryover
than the data show, and this difference is one of its approximations.

\begin{figure}[ht]
\centering
\includegraphics[width=\textwidth]{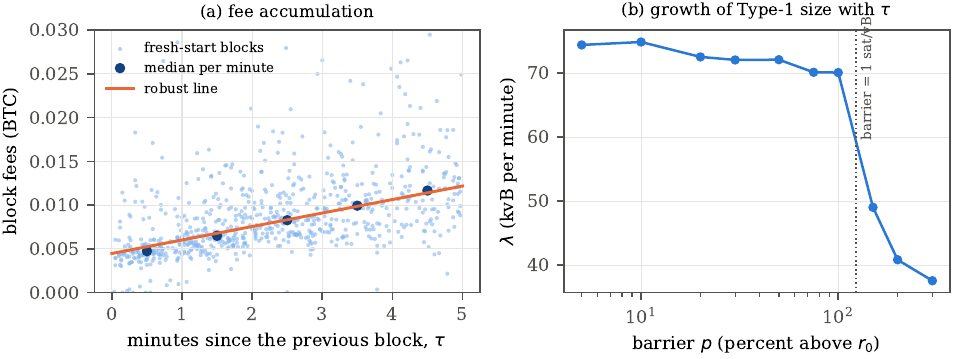}
\caption{Estimating the fee inputs from the 737 blocks whose predecessor was found within
2 minutes. (a) Block fees against the time since the previous block, with the median per
minute and the robust line. (b) Growth of the Type-1 size of a block per minute, when a
transaction counts as Type-1 above the barrier $r_0(1+p)$.}
\label{fig:ec-fees}
\end{figure}
\end{document}